\documentclass[9pt]{IEEEtran}



\usepackage{cite}
\usepackage{amsmath,amssymb,amsfonts,amsthm}
\usepackage{graphicx}

\usepackage{mathrsfs}

\usepackage[dvipsnames]{xcolor}
\usepackage[colorlinks=true,allcolors=MidnightBlue]{hyperref}

\usepackage{float}
\usepackage{color}
\usepackage{multirow}
\usepackage{booktabs}
\usepackage{enumitem}
\usepackage{comment}

\usepackage[bb=boondox]{mathalfa}

\usepackage{subfigure}

\usepackage{threeparttable}

\usepackage{flushend}

\newcommand{\bmat}[1]{\begin{bmatrix}#1\end{bmatrix}}
\newcommand{\smat}[1]{\left[\begin{smallmatrix}#1\end{smallmatrix}\right]}

\newcommand{\real}{\mathbb{R}}
\newcommand{\dotsS}{\rotatebox{0}{\tiny\dots}}
\newcommand{\vdotsS}{\rotatebox{90}{\tiny\dots}}

\newcommand{\cl}{\rm{cl}}

\DeclareMathOperator{\sat}{sat}
\DeclareMathOperator{\sign}{sign}
\DeclareMathOperator{\blkdiag}{blkdiag}
\DeclareMathOperator{\sinc}{sinc}

\newtheorem{assumption}{Assumption}
\newtheorem{problem}{Problem}
\newtheorem{theorem}{Theorem}

\newtheorem{lemma}{Lemma}

\newcounter{remarkcounter}
\newenvironment{remark}{\noindent\refstepcounter{remarkcounter}{\bfseries Remark~\arabic{remarkcounter}.}}{\hspace*{0pt}\hfill$\blacksquare$}

\newcounter{examplecounter}

\allowdisplaybreaks
\graphicspath{{./images/}}

\begin{document}


\title{\LARGE{Stabilization of Nonlinear Systems with State-Dependent Representation: \\ From Model-Based to Direct Data-Driven Control}}

\author{Lidong Li, Rui Huang, and Lin Zhao
\thanks{Lidong Li, Rui Huang, and Lin Zhao are with the Department of Electrical and Computer Engineering, National University of Singapore, Singapore 117583 (e-mail: \texttt{\{li.ld, ruihuang, elezhli\}@nus.edu.sg}).}
}

\maketitle

\begin{abstract}
This paper presents a framework for stabilizing nonlinear systems in state-dependent form. We first develop a model-based formulation that constructs a polytope enclosing the state-dependent system matrices over a prescribed state region and then synthesizes the controller by solving linear matrix inequalities (LMIs) at the polytope vertices. To establish stability and region of attraction (ROA) guarantees from the vertex LMIs, we develop a model-based mechanism and then apply it to direct data-driven control. In the data-driven setting, we use a known library of basis functions to represent the dynamics with unknown coefficients. A polytope captures the basis-function variation over a state region, and a data-consistent set characterizes the noise-induced coefficient uncertainty. This data-consistent set is incorporated directly into the vertex LMIs via Petersen's lemma. Building on our model-based mechanism, the resulting vertex LMIs yield a controller that guarantees stability and provides a certified ROA. The data-driven formulation further establishes disturbance-robustness guarantees and ROA estimates under input saturation. Numerical simulations and physical experiments demonstrate the effectiveness of the proposed controller.
\end{abstract}



\section{Introduction}

The analysis and control of nonlinear systems are challenging, especially when their mathematical models are unavailable.
In such cases, data-based control frameworks provide a way to reduce reliance on system models.
One framework is indirect: it first identifies the model from data and then applies model-based control. 
Another is direct data-driven control, which synthesizes controllers directly from data without explicit identification.
Examples include adaptive control~\cite{Astolfi2019Adaptive}, virtual reference feedback tuning~\cite{Campi2006VRFT}, iterative learning control~\cite{moore1993iterative}, and reinforcement learning~\cite{Lewis2018RL}. For comprehensive overviews, see~\cite{HOU2013Survey}.
Among direct data-driven approaches, those built on \emph{Willems et al.'s fundamental lemma}~\cite{WILLEMS2005325} have attracted growing interest and been extensively studied for linear time-invariant (LTI) systems~\cite{Persis2020, Henk2020Informativity, Lidong2024IO, Julian2020RobustDD, Florian2021Reviews}. 
These approaches can provide guarantees on stability, robustness, optimality, and safety.
However, extending these results remains challenging, as such guarantees do not readily carry over to nonlinear dynamics.

\noindent\emph{Related Work:}

Using Willems et al.'s fundamental lemma, \cite{Persis2020} represents LTI systems via finite input-state data and obtains controllers by solving data-dependent LMIs.
This approach is extended to nonlinear systems by linearizing the dynamics to first order and treating the nonlinear remainder as process noise, yielding local stability guarantees.
However, it requires restrictive assumptions on the remainder and does not estimate ROA.
The works in~\cite{Guo2022SOS, AndreaPetersen2022, huang2025SOS, Sznaier2021SOS} address polynomial systems and design globally asymptotically stabilizing controllers using sum-of-squares (SOS) programming. However, SOS methods often suffer from poor scalability, high computational cost, and numerical sensitivity.
Other approaches for more general nonlinear systems use function approximations such as polynomials~\cite{Guo2023Taylor, Martin2024SOS, Rapisarda202poly}, Gaussian processes~\cite{Christian2021statistical}, and kernel regression~\cite{Henk2023Kernel, Zhongjie2023Kernel}.
These approximations are often local~\cite{Martin2023survey}, and the resulting controllers can be conservative due to approximation errors and data noise. 
The Koopman operator provides a global linear representation of nonlinear systems, but its effectiveness depends critically on the choice of observables and the accuracy of finite-dimensional approximations~\cite{BEVANDA2021Reviews}, which can limit stability and robustness guarantees for strongly nonlinear or high-dimensional systems~\cite{Robin2025Koopman3, Robin2025Koopman1}.
Other recent works assume known basis-function libraries, thereby avoiding approximation-error analysis.
Such basis-function libraries are used for polynomial systems~\cite{Sznaier2021SOS}, data-driven feedback linearization~\cite{Lucas2021Linearizable, Mohammad2023Linearizable, Claudio2024Linearizable}, and nonlinearity cancellation~\cite{Claudio2023Cancel, Persis2023nonlinear, Abolfazl2025Nonlinear, Nima2025Versatile}.
The above methods synthesize controllers offline from batch data.
In contrast, \cite{Dai2021SDRE} proposes an online scheme that solves a state-dependent Riccati equation through data-based LMIs at each control step.
Yet, recursive feasibility of the online LMIs is not guaranteed.
This issue was partially addressed for noise-free data in~\cite{Xiaoyan2025Online}.
Since nonlinear stability guarantees are often local, several works have studied data-driven ROA estimation: \cite{Ferrante2026ROA} considers linear systems with input saturation, \cite{Tacchi2025ROA} learns piecewise-affine Lyapunov for ROA estimation, and \cite{Alsalti2025Descriptor} estimates the ROA via descriptor embedding.

Despite extensive research over the years, direct data-driven nonlinear control still remains an open problem.
The key challenges include designing controllers with guaranteed stability and estimating the ROA from finite data samples.
Furthermore, data noise is inevitable and induces uncertainties. 
The control law should ensure robustness to such uncertainties.
Other closed-loop properties, such as input-to-state stability~\cite[\S 7]{Nonlinear2023Kellett}\cite{JIANG2001ISS} and region of attraction subject to input saturation~\cite[\S 8.2]{Nonlinear2023Kellett}\cite{Sophie2011Saturation}, are important yet remain underexplored in data-driven nonlinear control with unknown model parameters.

\noindent\emph{Contributions:}

We first develop a model-based formulation as a basis for the subsequent data-driven controller synthesis.
The state-dependent system matrices over a prescribed state region are contained in a polytope constructed from element-wise bounds on their entries.
The controller is then obtained by solving LMIs at the polytope vertices.
Since this polytopic description is valid only over the prescribed state region, conventional stability analysis for parameter-varying systems \cite[\S 5]{LMI2000Scherer} is not directly applicable.
We thus establish a mechanism to guarantee stability and provide a certified ROA from the vertex LMIs.
We also investigate the closed-loop disturbance robustness and ROA estimation under input saturation.

We apply the model-based formulation to direct data-driven controller synthesis.
A library of known basis functions is assumed, and the state-dependent dynamics are written as a linear combination of these basis functions with unknown coefficients.
Over a prescribed state region, variation of the basis functions is captured by a polytope, and the noise-induced coefficient uncertainty is characterized by a data-consistent set.
The data-consistent set is directly incorporated into the LMIs at the polytope vertices via Petersen's lemma. 
Building on our model-based mechanism, solving these LMIs yields a controller that guarantees exponential stability and provides a certified ROA for all coefficients consistent with the data.
Robustness against disturbances and ROA estimation under input saturation are also established directly from data.

Our main contributions are summarized as follows.
\begin{itemize}[nosep,noitemsep,left=0pt]
    \item Building on our model-based mechanism for stability and ROA analysis, we develop a direct data-driven controller synthesis method for nonlinear systems in state-dependent form.
    The state-dependent variation and noise-induced coefficient uncertainty are characterized separately, and are jointly incorporated into the controller synthesis.
    Beyond stability, the proposed method also provides disturbance-robustness guarantees and ROA estimates under input saturation, which are rarely addressed in existing data-driven nonlinear control methods.
    \item Our framework applies to a broader class of nonlinear systems than feedback linearization or nonlinearity cancellation methods. Unlike state-dependent Riccati equation approaches that solve the online optimization problem at each control step, our controller is synthesized offline, avoiding the issue of recursive infeasibility during online operation. Numerical comparisons also show that our data-dependent LMIs exhibit better computational scalability than SOS-based methods in the considered benchmark.
\end{itemize}

\noindent\emph{Notations and Definition:}

The sets of real and natural numbers are denoted by $\mathbb{R}$ and $\mathbb{N}$, respectively.
The identity and zero matrices are denoted by $I_n$ and $0_{m \times n}$, with the subscripts omitted when clear from context.
Let $(v_1,\dots,v_n)$ denote the stacked vector $[v_1^\top \ \dots \ v_n^\top]^\top$.
For a symmetric matrix $M$, $\lambda_{\max}(M)$ and $\lambda_{\min}(M)$ denote its largest and smallest eigenvalues.
The 2-norm of a vector $v$ and the induced 2-norm of a matrix $M$ are denoted by $|v|$ and $\|M\|$.
For matrices $M$, $N$, and $O$, 
$MN \cdot O [\ast]^{\top}:=MNO(MN)^{\top}$ and 
$[\ast]^{\top} O \cdot MN:=(MN)^{\top}OMN$.
For a symmetric matrix $\left[\begin{smallmatrix} M & N \\ N^{\top} & O \end{smallmatrix}\right]$, we may write it as $\left[\begin{smallmatrix} M & N \\ \ast & O \end{smallmatrix}\right]$ or $\left[\begin{smallmatrix} M & \ast \\ N^{\top} & O \end{smallmatrix}\right]$.
For a matrix $M$, $M_{(i)}$ denotes its $i$-th row.
For a signal $\{h(k)\}_{k=0}^{\infty}$, $|h|_{\infty}:=\sup_{k\geq0}|h(k)|$.
The classes $\mathcal{K}$ and $\mathcal{KL}$ follow \cite[Defs. 1.6 and 1.9]{Nonlinear2023Kellett}.
The operator $\blkdiag(\cdot)$ denotes block-diagonal concatenation.
We denote the ball and ellipsoid by $\mathcal{B}_{r}:=\{x\in\real^{n}:|x|\leq r\}$ and $\mathcal{E}_{M,\tau}:=\{x\in\real^{n}:x^{\top}Mx\leq\tau\}$, respectively.
A set $\mathcal{R}\subset\mathbb{R}^n$ is called a region of attraction (ROA) for the equilibrium $x=0$ of the system $x^{+}=f(x)$ if $x(k)\to0$ as $k\to\infty$ for every $x(0)\in\mathcal{R}$ \cite[Def. 2.28]{Nonlinear2023Kellett}.

\section{Model-Based Controller}\label{sec:ModelBased}

In this section, we develop a model-based formulation as a basis for the subsequent direct data-driven controller synthesis. In particular, the stability and ROA analysis established here will be applied to the data-driven setting in Sections~\ref{sec:DBcontroller}--\ref{sec:InputSatur_data}.

\vspace{-1mm}
\subsection{Problem Formulation}\label{sec:ModelBasedPri}
Consider the discrete-time nonlinear system in state-dependent form
\begin{equation}\label{eq:nSYS}
    x^{+} = A(x)x + B(x) u,
\end{equation}
where $x \in \real^{n}$ is the state; 
$x^{+}$ is the forward shifting, i.e., $x^{+}(k) := x(k+1)$ for $k \in \mathbb{N}$;
$u \in \real^{m}$ is the control input.
For the state dependent matrices 
\vspace{-2mm}
\begin{equation}\label{eq:entries}
    A(x) = [a_{ij}(x)]_{n \times n}
    \ \text{ and } \
    B(x) = [b_{ij}(x)]_{n \times m},
\end{equation}
we assume that the entries $a_{ij} : \mathbb{X} \mapsto \real$ and $b_{ij} : \mathbb{X} \mapsto \real$ 
are continuous functions on the domain $\mathbb{X} \subseteq \real^{n}$, and $\mathbb{X}$ should contain the origin $\{0\}$.

The problem of this section is as follows.

\begin{problem}\label{problem:modelBased}
Design a state feedback controller $u = Kx$ such that 
the origin of the closed-loop system $x^{+} = \big( A(x) + B(x)K \big) x$ is locally exponentially stable.  
\end{problem}

\vspace{-4mm}
\subsection{Model-Based Controller Design}\label{sec:MBcontroller}

The following lemma lays the groundwork for solving Problem~\ref{problem:modelBased}.

\begin{lemma}\label{lema:modelBasedSolution}
Given a positive real number $r$, if there exist $P \succ 0$, $\epsilon > 0$ and $K$ such that
\begin{subequations}\label{eq:LyaCondition}
\begin{align}
    & \big[ A(x) + B(x)K \big]^{\top} P \big[ A(x) + B(x)K \big] - P \preceq - \epsilon I_n \label{eq:LyaCondition1} \\
    & \hspace{11mm} \forall x \in \mathcal{B}_{r} := \{ x \in \real^{n} : |x| \leq r \} \subseteq \mathbb{X},
\end{align}
\end{subequations}
then there exists a region of attraction 
\vspace{-1mm}
\begin{align}
    \mathcal{B}_{r_0} \! := \! \{ x \in \real^{n} : |x| \leq r_0 \}, 
    \
    r_0 \! = \! \min\left\{ r, \! \sqrt{ \tfrac{\lambda_{\min}(P)} {\lambda_{\max}(P)-\epsilon} } \cdot r \right\} \!, 
    \label{eq:modelBased_r0}
\end{align}
and the closed-loop system $x^{+} = \big( A(x) + B(x)K \big) x$
starting from any $x(0) \in \mathcal{B}_{r_0}$ 
will exponentially converge to the origin.
\end{lemma}

\begin{proof}
Let $V(x)=x^\top Px$.
Condition~\eqref{eq:LyaCondition} yields
$V(x(k+1))\leq \mu V(x(k))$ for all $x(k)\in\mathcal B_r$, with some $0<\mu<1$.
The choice of $r_0$ guarantees that trajectories starting from $\mathcal B_{r_0}$ remain in $\mathcal B_r$, and hence the above inequality holds for all $k\geq0$, implying exponential convergence to the origin.
See Appendix~\ref{apx:lema:modelBasedSolution} for details.
\end{proof}

The condition \eqref{eq:LyaCondition1} implies $P-\epsilon I_n \succeq 0$ and hence $\epsilon \leq \lambda_{\min}(P)$. Since \eqref{eq:LyaCondition1} remains valid for any smaller positive $\epsilon$, without loss of generality, we take $0 < \epsilon < \lambda_{\min}(P) \leq \lambda_{\max}(P)$. Therefore, $\lambda_{\max}(P)-\epsilon>0$, and the denominator in \eqref{eq:modelBased_r0} is positive.
The same argument applies to the denominators in all subsequent ROA-radius expressions, namely, \eqref{eq:LMI_r0}, \eqref{eq:modelBased_r0_satur}, \eqref{eq:modelBased_r0_satur_thm}, \eqref{eq:LMI_r0_data}, \eqref{eq:LMI_r0_data_bar}, and \eqref{eq:modelBased_r0_satur_thm_data}.

It is difficult to directly solve \eqref{eq:LyaCondition}, since $\mathcal{B}_r$ contains infinitely many elements. To overcome this difficulty, we construct a convex polytope containing $G(x):=[A(x)\ B(x)]$ for all $x\in\mathcal{B}_r$.
Specifically, define
\begin{equation}\label{eq:GxOO}
    \mathcal{G} := \big\{ G \in \real^{n \times (n+m)} : G = [A(x) \ B(x)], \ x \in \mathcal{B}_r \big\}.
\end{equation}
Bearing in mind \eqref{eq:entries}, let
\begin{subequations}\label{eq:min_max_gji}
\begin{align}
    & \underline{a}_{ij} := \min_{x \in \mathcal{B}_r} a_{ij}(x), & 
    & \bar{a}_{ij} := \max_{x \in \mathcal{B}_r} a_{ij}(x),
    \\
    & \underline{b}_{ij} := \min_{x \in \mathcal{B}_r} b_{ij}(x), & 
    & \bar{b}_{ij} := \max_{x \in \mathcal{B}_r} b_{ij}(x). 
\end{align}
\end{subequations}
Note that the maximum and minimum values exist, 
since $a_{ij}(x)$ and $b_{ij}(x)$ are assumed to be continuous on the domain $\mathbb{X} \subseteq \real^{n}$,
and $\mathcal{B}_r \subseteq \mathbb{X}$ is defined as a compact set.
Moreover, define
\begin{align}\label{eq:bar_set_G}
    & \bar{\mathcal{G}} := 
    \bigg\{ \Big[ [\alpha_{ij}]_{n \times n} \ \ [\beta_{ij}]_{n \times m} \Big] \in \real^{n \times (n+m)} : 
    \\[-6pt]   
    & \hspace{35mm} \alpha_{ij} \in \{ \underline{a}_{ij}, \ \bar{a}_{ij} \}  ,  \
    \beta_{ij} \in \{ \underline{b}_{ij}, \ \bar{b}_{ij} \}  \bigg\} . \notag
\end{align}
By the definition and property of convex hull ${\rm{conv}}(\cdot)$ in \cite[Def. 1.5, Prop. 1.6]{LMI2000Scherer}, we have
\begin{align*}
    & {\rm{conv}}(\bar{\mathcal{G}}) := 
    \bigg\{ \Big[ [\alpha_{ij}]_{n \times n} \ \ [\beta_{ij}]_{n \times m} \Big] \in \real^{n \times (n+m)} : 
    \\[-6pt]   
    & \hspace{35mm} \alpha_{ij} \in [ \underline{a}_{ij} , \ \bar{a}_{ij} ]  ,  \
    \beta_{ij} \in [ \underline{b}_{ij} , \ \bar{b}_{ij} ]  \bigg\} . \notag
\end{align*}
It is obvious that 
\vspace{-1mm}
\begin{equation}\label{eq:G_subset_convG}
    \mathcal{G} \subseteq {\rm{conv}}(\bar{\mathcal{G}}).
\end{equation}
Since the maximum and minimum values are computed independently for entries $a_{ij}$ and $b_{ij}$, $\operatorname{conv}(\bar{\mathcal G})$ may contain elements that cannot be generated by any $x\in\mathcal B_r$, which introduces conservatism. This conservatism can be large when entries have large variations within $\mathcal B_r$, while it can remain small when entries are constant or vary slightly. 
From a computational perspective, the set $\bar{\mathcal G}$ contains at most $2^{n(n+m)}$ vertices, with this maximum attained only when the upper and lower bounds differ for every entry of $A$ and $B$. In practice, these matrices often contain many constant entries that do not contribute additional vertices. Thus, the actual number of distinct vertices in $\bar{\mathcal G}$ can be much smaller than $2^{n(n+m)}$.
Now with the result of \eqref{eq:G_subset_convG}, 
we give a solvable sufficient condition to \eqref{eq:LyaCondition}.

\begin{theorem}\label{thm:LMIsolution}
Given a positive real number $r$ and the set $\mathcal{B}_r$, 
if there exist $\Gamma \succ 0$, $\epsilon_{\Gamma} > 0$ and $Y$ such that
\begin{equation}\label{eq:LMI_SufficientLyaCondition}
    \bmat{ 
        -\Gamma & G \smat{ \Gamma \\ Y} \\
        [\Gamma \ Y^{\top} ] G^{\top} & - \Gamma + \epsilon_{\Gamma} I_n
    } \preceq 0
    \quad \forall G \in \bar{\mathcal{G}},
\end{equation}
and let $u = Kx$ with $K = Y \Gamma^{-1}$, 
then there exists a region of attraction $\mathcal{B}_{r_0}$ with
\begin{equation}\label{eq:LMI_r0}
    r_0 = \min \left\{ r, \ \sqrt{ \tfrac{ \lambda_{\max}(\Gamma) \lambda_{\min}(\Gamma) }{ \lambda^2_{\max}(\Gamma) - \epsilon_{\Gamma} \lambda_{\min}(\Gamma) } } \cdot r \right\},
\end{equation}
and the closed-loop system $x^{+} = \big( A(x) + B(x)K \big) x$  
starting from any $x(0) \in \mathcal{B}_{r_0}$ 
will exponentially converge to the origin.
\end{theorem}

\begin{proof}
By affinity, \eqref{eq:LMI_SufficientLyaCondition} at all vertices of $\bar{\mathcal{G}}$ implies the corresponding condition for every $G\in\operatorname{conv}(\bar{\mathcal{G}})$, and hence, by \eqref{eq:G_subset_convG}, for every $G\in\mathcal{G}$.
Using the Schur complement with $K=Y\Gamma^{-1}$ and $P=\Gamma^{-1}$ then yields \eqref{eq:LyaCondition}, so Lemma~\ref{lema:modelBasedSolution} applies and gives the ROA in \eqref{eq:LMI_r0}.
See Appendix~\ref{apx:thm:LMIsolution} for details.
\end{proof}

\vspace{-4mm}
\subsection{Robustness Against Disturbances}\label{sec:robustness}

In this subsection, we explore the robustness of the controller against disturbances. 
To be specific, during the controller implementation, we consider the following closed-loop model
\begin{equation}\label{eq:disturbed_ClosedLoop}
    x^{+} = \big( A(x) + B(x)K \big) x + w
\end{equation}
where $K$ is solved from LMIs \eqref{eq:LMI_SufficientLyaCondition} 
and $w \in \real^{n}$ is the disturbance.
It is proved in Theorem~\ref{thm:LMIsolution} that, 
without disturbance ($w \equiv 0$), 
the closed-loop system is exponentially stable if initial state $x(0) \in \mathcal{B}_{r_0}$.
While under the influence of the disturbance $w$, 
the system exhibits a property similar to the input-to-state stability \cite{JIANG2001ISS},
which is established in the next theorem.

\begin{theorem}\label{thm:Robust2Dist}
For the disturbed closed-loop system \eqref{eq:disturbed_ClosedLoop}, given a positive real number $r$ and the set $\mathcal{B}_r$, 
suppose that there exist $\Gamma \succ 0$, $\epsilon_{\Gamma} > 0$ and $Y$ such that \eqref{eq:LMI_SufficientLyaCondition} holds, 
and let $K = Y \Gamma^{-1}$ for \eqref{eq:disturbed_ClosedLoop}.
If the initial state $x(0)$ and disturbance $w(k)$ satisfy
\begin{subequations}\label{eq:R2D_condition}
\begin{align}
    |x(0)|^2 \leq r^2 \,
    \text{ and } \ \delta_{x_0}  |x(0)|^2 + \delta_{w}  |w|^2_{\infty} \leq r^2 ,
\end{align}
where 
\begin{align}
    \delta_{x_0} & := \frac{ 2 \lambda^2_{\max}(\Gamma) - \epsilon_{\Gamma} \lambda_{\min}(\Gamma) } { 2 \lambda_{\min}(\Gamma) \lambda_{\max}(\Gamma) } , \\
    \delta_{w} & := \frac{ 4 \gamma^2_{A_{\cl}} \lambda^5_{\max}(\Gamma) + 2 \epsilon_{\Gamma} \lambda_{\min}(\Gamma) \lambda^3_{\max}(\Gamma) } { \epsilon^2_{\Gamma} \lambda^3_{\min}(\Gamma) } , \\
    \gamma_{A_{\cl}} &:= \max_{x \in \mathcal{B}_r} \| A(x) + B(x)K \| ,
\end{align}
\end{subequations}
then there exists $\beta \in \mathcal{KL}$ and $\gamma \in \mathcal{K}$ such that
\begin{equation}\label{eq:R2D_result}
    |x(k)| \leq \beta\big(  |x(0)|  , \, k \big) + \gamma(|w|_{\infty}) \quad \forall k \in \mathbb{N}.
\end{equation}
\end{theorem}

\begin{proof}
By Theorem~\ref{thm:LMIsolution}, feasibility of \eqref{eq:LMI_SufficientLyaCondition} provides a Lyapunov decrease condition over $\mathcal{B}_r$.
Applying it to \eqref{eq:disturbed_ClosedLoop} and using Young's inequality yields a recursive bound on the Lyapunov function $V$, while \eqref{eq:R2D_condition} guarantees that the state trajectory of \eqref{eq:disturbed_ClosedLoop} remains in $\mathcal{B}_r$.
This leads to \eqref{eq:R2D_result}.
See Appendix~\ref{apx:thm:Robust2Dist} for details.
\end{proof}

\subsection{System Subject to Input Saturation}\label{sec:InputSatur}

We consider the system subject to input saturation:
\begin{equation}\label{eq:nSYS_satur}
    x^{+} = A(x)x + B(x) \sat(u).
\end{equation}
The saturation function $\sat:\real^m \! \to \! \real^m$ is defined as
\begin{equation}
    \sat(u_i) := \sign(u_i) \min( |u_i|, \bar{u}_i )
    \ \text{ for } \ i=1,\cdots,m,
\end{equation}
where $u_i$ is the $i$-th control input and $\bar{u}_i>0$ is the $i$-th entry of the saturation-level vector $\bar u$.
Given the control law $u = Kx$,
the closed-loop system subject to input saturation is 
\begin{equation}\label{eq:nSYS_satur_closed}
    x^{+} = \big( A(x) + B(x)K \big) x + B(x) \phi(Kx) ,
\end{equation}
where $\phi(u)$ is the dead-zone function defined as
\begin{equation}\label{eq:dead_zone}
    \phi(u) := \sat(u) - u.
\end{equation}
Following \cite[Lem.~1.6]{Sophie2011Saturation} and \cite[Lem.~1]{Sophie2024Saturation},
the generalized sector lemma, as stated below, 
is used to handle the dead-zone \eqref{eq:dead_zone}.

\begin{lemma}\label{lema:sectorCondi}
Given a matrix $L \in \real^{m \times n}$, 
for every $x \in \mathcal{S}(L)$ where
\begin{equation}\label{eq:setSL}
    \mathcal{S}(L) := \left\{ x \in \real^{n} : 
    | L_{(i)} x | \leq \bar{u}_i, \
    \forall i = 1,\dots,m \right\},
\end{equation}
the following relation holds:
\begin{equation}\label{eq:sectorCondi}
    \phi(u)^{\top} N \big( \sat(u) + Lx \big) \leq 0
\end{equation}
with any diagonal positive definite matrix $N \in \real^{m \times m}$.
\end{lemma}

The following lemma is used to establish the main result.

\begin{lemma}\label{lema:intermeLyaCondi}
Given a positive real number $r$ and the set $\mathcal{B}_r$, 
if there exist $P \succ 0$, $\tau > 0$, $\epsilon > 0$, and $V(x) := x^{\top} P x$,
such that
\begin{equation}\label{eq:LyaCondition_satur}
    V(x^+) - V(x) \leq - \epsilon | x |^2 \quad
    \forall x \in \mathcal{B}_r \cap \mathcal{E}_{P,\tau} ,
\end{equation}
then there exists a region of attraction 
$\mathcal{B}_{r_0} \cap \mathcal{E}_{P, \tau}$
with
\begin{equation}\label{eq:modelBased_r0_satur}
    r_0 = \min\left\{ r, \ \sqrt{ \tfrac{\lambda_{\min}(P)} {\lambda_{\max}(P)-\epsilon} } \cdot r \right\},
\end{equation}
and the closed-loop system $x^{+} = A(x)x + B(x)\sat(Kx)$ 
starting from any $x(0) \in \mathcal{B}_{r_0} \cap \mathcal{E}_{P, \tau}$ 
will exponentially converge to the origin.
\end{lemma}

\begin{proof}
Condition~\eqref{eq:LyaCondition_satur} ensures that $V$ decreases over $\mathcal{B}_r\cap\mathcal{E}_{P,\tau}$.
By Lemma~\ref{lema:modelBasedSolution}, \eqref{eq:modelBased_r0_satur} guarantees that the closed-loop system trajectories starting from $\mathcal{B}_{r_0}$ remain in $\mathcal{B}_r$. The decrease of $V$ ensures that $\mathcal{E}_{P,\tau}$ is positively invariant.
Hence, trajectories starting from $\mathcal{B}_{r_0}\cap\mathcal{E}_{P,\tau}$ remain in $\mathcal{B}_r\cap\mathcal{E}_{P,\tau}$ and converge exponentially to the origin.
See Appendix~\ref{apx:lema:intermeLyaCondi} for details.
\end{proof}

The main result of this subsection is given below.

\begin{theorem}\label{thm:LMIsolution_satur}
Given a positive real number $r$ and the set $\mathcal{B}_r$,
if there exist $\Gamma \succ 0$, $\epsilon_{\Gamma} > 0$, $Y$, $W$ and diagonal positive definite matrix $S$ such that
\begin{subequations}\label{eq:LMI_SufficientLyaCondition_satur}
\begin{align}
    & \bmat{\Gamma & W_{(i)}^{\top} \\[3pt] W_{(i)} & \bar{u}_i^2} \succeq 0
    \quad i = 1,\dots,m,
    \label{eq:subsetCondi_satur}
    \\
    & \bmat{
        -\Gamma \! + \epsilon_{\Gamma} I_n 
        & \! -Y^{\top} \!\! - \! W^{\top} 
        & \! \smat{\Gamma \\ Y}^{ \!\! \top} \!\! G^{\top}
        \\[8pt]
        -Y- W & -2 S & \smat{0 \\ S}^{ \!\! \top} \!\! G^{\top} 
        \\[5pt]
        G \smat{\Gamma \\ Y} & G \smat{0 \\ S} & -\Gamma 
    } \! \preceq 0 
    \ \ \forall G  \in  \bar{\mathcal{G}},
    \label{eq:stableCondi_satur}  
\end{align}
\end{subequations}
and let $u = Kx$ with $K = Y \Gamma^{-1}$, 
then there exists a region of attraction 
$\mathcal{B}_{r_0} \cap \mathcal{E}_{P,1}$ with $P = \Gamma^{-1}$ and
\begin{equation}\label{eq:modelBased_r0_satur_thm}
    r_0 = \min \left\{ r, \ \sqrt{ 
        \tfrac{ \lambda_{\max}(\Gamma) \lambda_{\min}(\Gamma) }
        { \lambda^2_{\max}(\Gamma) - \epsilon_{\Gamma} \lambda_{\min}(\Gamma) }
    } \cdot r \right\} , 
\end{equation}
and the closed-loop system $x^{+} = A(x)x + B(x)\sat(Kx)$ 
starting from any $x(0) \in \mathcal{B}_{r_0} \cap \mathcal{E}_{P,1}$ 
will exponentially converge to the origin.
\end{theorem}

\begin{proof}
Condition~\eqref{eq:subsetCondi_satur} ensures that $\mathcal{E}_{P,1}$ lies inside the sector region of Lemma~\ref{lema:sectorCondi}. 
Meanwhile, by the same polytopic argument as in Theorem~\ref{thm:LMIsolution}, condition~\eqref{eq:stableCondi_satur} yields a Lyapunov decrease inequality with an additional dead-zone term over $\mathcal{B}_r$.
Combining the two gives \eqref{eq:LyaCondition_satur} on $\mathcal{B}_r\cap\mathcal{E}_{P,1}$, and Lemma~\ref{lema:intermeLyaCondi} then yields the stated ROA.
See Appendix~\ref{apx:thm:LMIsolution_satur} for details.
\end{proof}

\section{Data-Driven Controller}\label{sec:DataBased}

\subsection{Problem Formulation}\label{sec:DataBasedPri}
In this section, for the system \eqref{eq:nSYS}, 
while the state matrix $A(x)$ and input matrix $B(x)$ are assumed to be unknown,
we assume the existence of a known function library that can construct $A(x)$ and $B(x)$, as stated in the following assumption.

\begin{assumption}\label{ass:function_qi}
Condition (i) and (ii) below holds. 
\\
(i) For each $j=1,\dots,n$, we know a vector-valued function 
$\xi_{A_j} : \real^{n} \mapsto \real^{n_{A_j}}$ 
which is continuous on the domain $\mathbb{X} \subseteq \real^{n}$,
such that
\begin{equation}\label{eq:function_Aj}
    \big( a_{1j}(x), \dots, a_{ij}(x), \dots, a_{nj}(x) \big) = E_{A_j} \xi_{A_j}(x)
\end{equation}
for some constant matrix $E_{A_j} \in \real^{n \times n_{A_j}}$.
\\
(ii) For each $j=1,\dots,m$, we know a vector-valued function
$\xi_{B_j} : \real^{n} \mapsto \real^{n_{B_j}}$ 
which is continuous on the domain $\mathbb{X} \subseteq \real^{n}$,
such that
\begin{equation}\label{eq:function_Bj}
    \big( b_{1j}(x), \dots, b_{ij}(x), \dots, b_{nj}(x) \big) = E_{B_j} \xi_{B_j}(x)
\end{equation}
for some constant matrix $E_{B_j} \in \real^{n \times n_{B_j}}$.
\end{assumption}

Assumption~\ref{ass:function_qi} allows us to rewrite the system \eqref{eq:nSYS} as
\begin{subequations}\label{eq:nSYSrewrite}
\begin{align}
    x^{+} & = A(x)x + B(x) u  = [ E_{A_1} \xi_{A_1}(x) \dots E_{A_n} \xi_{A_n}(x) ] x \notag
    \\
    & \hspace{30mm} + [ E_{B_1} \xi_{B_1}(x) \dots E_{B_m} \xi_{B_m}(x) ] u \notag
    \\
    & = E_A \Xi_A(x) x + E_B \Xi_B(x) u 
\end{align}
where
\vspace{-2mm}
\begin{align}
    & E_A := [ E_{A_1} \dots E_{A_n} ] \in \real^{n \times n_A},
    \quad n_A := \sum_{j=1}^{n} n_{A_j} ,
    \\[-3pt]
    & \Xi_A(x) := \blkdiag \big( \xi_{A_1}(x), \dots, \xi_{A_n}(x) \big),
    \\[-3pt]
    & E_B := [ E_{B_1} \dots E_{B_m} ] \in \real^{n \times n_B},
    \quad n_B := \sum_{j=1}^{m} n_{B_j} ,
    \\[-3pt]
    & \Xi_B(x) := \blkdiag \big( \xi_{B_1}(x), \dots, \xi_{B_m}(x) \big).
\end{align}
\end{subequations}
We note that Assumption~\ref{ass:function_qi} corresponds to a grey-box setting, where the model structure, characterized by $(\Xi_A(x),\Xi_B(x))$, is known but the coefficients $(E_A,E_B)$ are unknown. This setting is common in mechanical and electrical systems, among others, where the basis-function library can often be derived from first principles while the coefficient values are not known. The unknown coefficients are instead characterized using experimental data.

During the experiment, 
we consider that the system \eqref{eq:nSYSrewrite} 
is affected by process noise $d \in \real^{n}$, 
i.e.,
\vspace{-1mm}
\begin{equation}\label{eq:nSYSrewrite_distur}
    x_{d}^{+} = E_A \Xi_A(x_d) x_d + E_B \Xi_B(x_d) u_d + d.
\end{equation}
The experimental input data is $\{u_{d}(k)\}_{k=0}^{T-1}$ 
and the corresponding state data is $\{x_{d}(k)\}_{k=0}^{T}$, 
where $T>0$ is the number of samples.

The problem of this section is summarized as follows.

\begin{problem}\label{problem:dataBased}
Without the model knowledge $( E_A, E_B )$,
using the data $\{u_{d}(k)\}_{k=0}^{T-1}$ and $\{x_{d}(k)\}_{k=0}^{T}$, 
design a state feedback controller $u = Kx$ such that 
the origin of the closed-loop system $x^{+} = ( A(x) + B(x)K ) x$ is locally exponentially stable.  
\end{problem}

\subsection{Data-Driven Controller Design}\label{sec:DBcontroller}
We first arrange the data as the following matrices
\begin{subequations}\label{eq:dataMatrix}
\begin{align}
    X_0 & \! := \big[ \Xi_A(x_d(0)) x_d(0) \dots \Xi_A(x_d(T \! - \! 1)) x_d(T \! - \! 1) \big] , \\
    X_1 & \! := \big[ x_d(1) \dots x_d(T) \big] , \\
    U_0 & \! := \big[ \Xi_B(x_d(0)) u_d(0) \dots \Xi_B(x_d(T \! - \! 1)) u_d(T \! - \! 1) \big] ,
\end{align}
and arrange the unknown noise as
\begin{align}
    D_0 :=  \big[ d(0) \dots d(T-1) \big] . \label{eq:dataMatrix_D0}
\end{align}
\end{subequations}
According to \eqref{eq:nSYSrewrite_distur}, 
the above data matrices satisfy
\begin{equation}\label{eq:DataRelation}
    X_1 = E_A X_0 + E_B U_0 + D_0.
\end{equation}
For further analysis, we assume that the noise sequence $D_0$ has bounded energy,
which is described in the following assumption.

\begin{assumption}\label{ass:DisturbanceEnergy}
The energy bound $\Theta \succeq 0$ is known, and
\begin{equation}\label{eq:DisturbanceEnergy}
    D_0 \in \mathcal{D} := \{ D \in \real^{n \times T} : D D^{\top} \preceq \Theta \}.
\end{equation}
\end{assumption}

In many practical scenarios, the noise energy bound $\Theta$ can be obtained or reasonably estimated from sensor specifications, prior experiments, or known disturbance bounds. For example, if each component of the noise satisfies $|d_i(k)|\leq \bar d$, then $D_0D_0^\top \preceq nT\bar d^2 I_n$, so that one can choose $\Theta=nT\bar d^2 I_n$. For Gaussian noise over a finite data horizon, suppose $d(k)\sim\mathcal{N}(0,\sigma^2I_n)$. Then, for any $\mu>0$, $D_0D_0^\top\preceq \sigma^2T(1+\mu+\sqrt{n/T})^2I_n$ with probability at least $1-2e^{-T\mu^2/2}$ \cite[Corollary~5.35]{vershynin2012nonasymptotic}.

Although the model knowledge of $E_A$ and $E_B$ is absent,
Assumption~\ref{ass:DisturbanceEnergy} enables us to characterize the set of all possible pairs $(\hat{E}_A, \hat{E}_B)$
that can generate data $\{u_{d}(k)\}_{k=0}^{T-1}$ and $\{x_{d}(k)\}_{k=0}^{T}$ 
while keeping the noise sequence in the set $\mathcal{D}$, namely,
\begin{subequations}\label{eq:set_C}
\begin{align}
    & \mathcal{C} 
    := \left\{ [ \hat{E}_A \ \hat{E}_B ] : X_1 = \hat{E}_A X_0 + \hat{E}_B U_0 + D, \ D \in \mathcal{D} \right\} \notag 
    \\
    & = \left\{ [ \hat{E}_A \ \hat{E}_B ] : \Big[ X_1 - \hat{E}_A X_0 - \hat{E}_B U_0 \Big] \cdot \Big[\ast\Big]^{\top}  \preceq \Theta \right\} \notag 
    \\
    & = \left\{ Z \! = \! [\hat{E}_A \ \hat{E}_B] : Z \mathbf{A} Z^{\top} \!\! + Z \mathbf{B}^{\top} \!\! + \mathbf{B} Z^{\top} \!\! + \mathbf{C} \preceq 0 \right\}  \label{eq:set_C_01}
\end{align}
with
\vspace{-2mm}
\begin{align}
    \mathbf{A} := \smat{X_0 \\ U_0} \smat{X_0 \\ U_0}^{\top} \!\!\! , \
    \mathbf{B} := - X_1 \smat{X_0 \\ U_0}^{\top} \!\!\! , \
    \mathbf{C} := X_1 X_1^{\top} - \Theta.  
\end{align}
\end{subequations}
It is clear that $[{E}_A \ {E}_B] \in \mathcal{C}$.
Compared with an identify-then-control approach, a main advantage of our framework is that the uncertainty induced by finite noisy data can be incorporated directly into robust controller synthesis. Specifically, combining the noise-energy bound $D_0D_0^\top\preceq\Theta$ with \eqref{eq:DataRelation} yields a direct characterization of the coefficient uncertainty, namely, the set $\mathcal C$, which preserves the correlations and anisotropic geometry of the uncertainty induced by the data. By contrast, identify-then-control approaches typically represent the identification error by a norm bound, which may not preserve the geometry of the data-induced uncertainty and can therefore introduce additional conservatism.
Accordingly, condition \eqref{eq:LyaCondition_data} later is to enforce stability for all systems in $\mathcal{C}$. Over-estimation of $\Theta$ enlarges $\mathcal{C}$, making the true system less likely to be excluded, but increasing conservatism by requiring stability over a larger uncertainty set. Under-estimation of $\Theta$ may improve SDP feasibility by shrinking $\mathcal{C}$, but if the actual noise no longer satisfies $D_0D_0^\top\preceq\Theta$, the true system may fall outside $\mathcal{C}$ and its stability is then not guaranteed. For the geometry of $\mathcal C$, see \cite{BISOFFI2021tradeoffs}.

We make the next assumption on data.
\begin{assumption}\label{ass:DataFullRow}
Data matrix $\smat{X_0 \\ U_0}$ has full row rank.
\end{assumption}

By Assumption~\ref{ass:DataFullRow}, the matrix $\mathbf A = [ X_0^\top \ U_0^\top ]^\top [ X_0^\top \ U_0^\top ]$ is positive definite, and hence $\mathbf A^{-1}$ and $\mathbf A^{-1/2}$ are well defined. Then, by \cite[Prop. 1]{AndreaPetersen2022}, the set $\mathcal{C}$ equivalently rewrites as
\begin{subequations}\label{eq:DataConsistentSetC_re}
\begin{align}
    & \mathcal{C}  = \big\{ \mathbf{Z}_c + \mathbf{Q}^{1/2} \Upsilon \mathbf{A}^{-1/2} \colon \Upsilon \Upsilon^\top \preceq  I_n \big\} , \\
    & \ \mathbf{Z}_c := - \mathbf{B} \mathbf{A}^{-1}, \ \mathbf{Q} := \mathbf{B} \mathbf{A}^{-1} \mathbf{B}^\top - \mathbf{C}. \label{eq:setC:ZQ_Gxu}
\end{align}
\end{subequations}

Lemma~\ref{lema:modelBasedSolution} in the model-based case lays the groundwork for solving Problem~\ref{problem:modelBased} and also provides the starting point for Problem~\ref{problem:dataBased}.
Specifically, the known matrices $[E_A\ E_B]$ in the model-based condition are replaced by $\mathcal C$.
We therefore seek $P\succ0$, $\epsilon>0$, and $K$ such that
\begin{subequations}\label{eq:LyaCondition_data}
\begin{align}
    & \big[ \ast \big]^{\top} P \cdot \big[ \hat{E}_A \Xi_A(x) + \hat{E}_B \Xi_B(x) K \big] - P 
    \preceq - \epsilon I_n \label{eq:LyaCondition1_data} 
    \\
    & \ \forall x \in \mathcal{B}_{r} 
    \text{ and } 
    \forall (\hat{E}_A, \hat{E}_B) \in \mathcal{C} . \label{eq:LyaCondition2_data}
\end{align}
\end{subequations}
However, in \eqref{eq:LyaCondition2_data}, 
there are infinite elements in $\mathcal{B}_{r}$ and $\mathcal{C}$. 
To overcome these challenges, 
we first use the convex polytope of 
$Q(x) := \blkdiag\big( \Xi_A(x), \Xi_B(x) \big)$ for $x \in \mathcal{B}_r$, 
and then apply Petersen's lemma \cite{AndreaPetersen2022} to handle $\mathcal{C}$.

Similar to \eqref{eq:GxOO}, we introduce
\begin{align}\label{eq:QxOO}
    & \mathcal{Q} := \big\{ Q \in \real^{(n_A + n_B) \times (n+m)} :
    \\[-3pt] 
    & \hspace{20mm}  Q(x) = \blkdiag\big( \Xi_A(x), \Xi_B(x) \big), \ x \in \mathcal{B}_r \big\}.
    \notag
\end{align}
Under Assumption~\ref{ass:function_qi}, 
the entries of $\xi_{A_j}$ for $j=1,\dots,n$ and $\xi_{B_j}$ for $j=1,\dots,m$ are denoted by
\vspace{-1mm}
\begin{equation*}
    \big( \xi^{A}_{1j}, \dotsS, \xi^{A}_{ij}, \dotsS, \xi^{A}_{n_{A_j}j} \big)
    := \xi_{A_j},
    \ \
    \big( \xi^{B}_{1j}, \dotsS, \xi^{B}_{ij}, \dotsS, \xi^{B}_{n_{B_j}j} \big)
    := \xi_{B_j}.
\end{equation*}
Let
\vspace{-2mm}
\begin{subequations}\label{eq:min_max_qji}
\begin{align}
    & \underline{\xi}^{A}_{ij} := \min_{x \in \mathcal{B}_r} \xi^{A}_{ij}(x), & %
    & \bar{\xi}^{A}_{ij} := \max_{x \in \mathcal{B}_r} \xi^{A}_{ij}(x) ,
    \\
    & \underline{\xi}^{B}_{ij} := \min_{x \in \mathcal{B}_r} \xi^{B}_{ij}(x), & %
    & \bar{\xi}^{B}_{ij} := \max_{x \in \mathcal{B}_r} \xi^{B}_{ij}(x) .
\end{align}    
\end{subequations}
Note that the maximum and minimum values exist, 
since $\xi_{A_j}(x)$ and $\xi_{B_j}(x)$ are continuous by Assumption~\ref{ass:function_qi}.
Moreover, define
\begin{align}\label{eq:bar_set_Q}
    \bar{\mathcal{Q}} := \bigg\{ 
    & \blkdiag \big( \zeta_{A_1},\dots, \zeta_{A_n}, \zeta_{B_1},\dots, \zeta_{B_m} \big) : 
    \\[-3pt]
    & \zeta_{A_j} := \left( \zeta^{A}_{1j}, \dots, \zeta^{A}_{ij}, \dots, \zeta^{A}_{n_{A_j} j} \right) \! ,
    \ \zeta^{A}_{ij} \in \!  \left\{ \underline{\xi}^{A}_{ij} , \ \bar{\xi}^{A}_{ij} \right\}  \notag
    \\
    & \zeta_{B_j} := \left( \zeta^{B}_{1j}, \dots, \zeta^{B}_{ij}, \dots, \zeta^{B}_{n_{B_j} j} \right) \! ,
    \ \zeta^{B}_{ij} \in \!  \left\{ \underline{\xi}^{B}_{ij} , \ \bar{\xi}^{B}_{ij} \right\}  
    \! \bigg\} . \notag
\end{align}
By the definition and property of convex hull, we have
\begin{align*}
    & {\rm{conv}}(\bar{\mathcal{Q}}) := \bigg\{ 
    \blkdiag \big( \zeta_{A_1},\dots, \zeta_{A_n}, \zeta_{B_1},\dots, \zeta_{B_m} \big) : 
    \\[-3pt]
    & \hspace{11mm} 
    \zeta_{A_j} := \left( \zeta^{A}_{1j}, \dots, \zeta^{A}_{ij}, \dots, \zeta^{A}_{n_{A_j} j} \right),
    \ \zeta^{A}_{ij} \in  \left[ \bar{\xi}^{A}_{ij}, \ \underline{\xi}^{A}_{ij} \right]  
    \\
    & \hspace{11mm} 
    \zeta_{B_j} := \left( \zeta^{B}_{1j}, \dots, \zeta^{B}_{ij}, \dots, \zeta^{B}_{n_{B_j} j} \right),
    \ \zeta^{B}_{ij} \in  \left[ \bar{\xi}^{B}_{ij}, \ \underline{\xi}^{B}_{ij} \right] 
    \! \bigg\} . 
\end{align*}
It is obvious that
\begin{equation}\label{eq:Q_subset_convQ}
    \mathcal{Q} \subseteq {\rm{conv}}(\bar{\mathcal{Q}}).
\end{equation}
Since the maximum and minimum values are computed independently for entries $\xi^{A}_{ij}$ and $\xi^{B}_{ij}$, $\operatorname{conv}(\bar{\mathcal Q})$ may contain elements that cannot be generated by any $x\in\mathcal B_r$, which introduces conservatism. This conservatism can be large when entries have large variations within $\mathcal B_r$, while it can remain small when entries are constant or vary slightly. 
From a computational perspective, the set $\bar{\mathcal{Q}}$ contains at most $2^{n_A+n_B}$ vertices, with this maximum attained only when the upper and lower bounds differ for every entry of $\Xi_A$ and $\Xi_B$. In practice, there are many constant entries that do not contribute additional vertices. Thus, the actual number of distinct vertices in $\bar{\mathcal Q}$ is typically much smaller than $2^{n_A+n_B}$.

Analogous to Theorem~\ref{thm:LMIsolution}, we have the next result.

\begin{lemma}\label{lema:polytope_data}
Under Assumption~\ref{ass:function_qi} and~\ref{ass:DisturbanceEnergy}, 
given a positive real number $r$ and the set $\mathcal{B}_r$, 
if there exist $\Gamma \succ 0$, $\epsilon_{\Gamma} > 0$ and $Y$ such that
\begin{subequations}\label{eq:LMI_SufficientLyaCondition_data}
\begin{align}
    & \bmat{ 
        -\Gamma & 
        [\hat{E}_A \ \hat{E}_B] Q \smat{\Gamma \\ Y}
        \\
        \smat{\Gamma \\ Y}^{\top} \!\! Q^{\top} [\hat{E}_A \ \hat{E}_B]^{\top}
        & - \Gamma + \epsilon_{\Gamma} I_n
    } \preceq 0
    \\
    & \hspace{2mm}
    \forall (\hat{E}_A,\hat{E}_B) \in \mathcal{C} \ \text{ and } \
    \forall Q \in \bar{\mathcal{Q}},
    \label{eq:LMI_SufficientLyaCondition_data-2}
\end{align}    
\end{subequations}
and let $u = Kx$ with $K = Y \Gamma^{-1}$, 
then there exists a region of attraction $\mathcal{B}_{r_0}$ with
\begin{equation}\label{eq:LMI_r0_data}
    r_0 = \min \left\{ r, \ \sqrt{ \tfrac{ \lambda_{\max}(\Gamma) \lambda_{\min}(\Gamma) }{ \lambda^2_{\max}(\Gamma) - \epsilon_{\Gamma} \lambda_{\min}(\Gamma) } } \cdot r \right\},
\end{equation}
and the closed-loop system $x^{+} = \big( A(x) + B(x)K \big) x$  
starting from any $x(0) \in \mathcal{B}_{r_0}$ 
will exponentially converge to the origin.
\end{lemma}

\begin{proof}
Assumption~\ref{ass:DisturbanceEnergy} guarantees that the true coefficient matrices satisfy $[E_A\ E_B]\in\mathcal{C}$.
Then, \eqref{eq:LMI_SufficientLyaCondition_data}, together with \eqref{eq:Q_subset_convQ} and \eqref{eq:nSYSrewrite}, reduces to the model-based condition in Theorem~\ref{thm:LMIsolution}, from which the result follows.
See Appendix~\ref{apx:lema:polytope_data} for details.
\end{proof}

The set $\mathcal{C}$ in \eqref{eq:LMI_SufficientLyaCondition_data-2} is handled using Petersen's lemma \cite{AndreaPetersen2022}.

\begin{lemma}\label{lema:petersen_data}
Under Assumption~\ref{ass:DataFullRow},
if there exist $\bar{\Gamma} \succ 0$, $\epsilon_{\bar{\Gamma}} > 0$ and $\bar{Y}$ such that
\begin{align}
    & \bmat{ 
        -\bar{\Gamma} - \mathbf{C} \!\! & \!\! 0 & \mathbf{B} 
        \\
        0 \!\!
        & \!\! - \bar{\Gamma} + \epsilon_{\bar{\Gamma}} I_n 
        & -\smat{\bar{\Gamma} \\ \bar{Y}}^{\top} \!\!\! Q^{\top} 
        \\[5pt]
        \mathbf{B}^{\top} \!\! & \!\! -Q \smat{\bar{\Gamma} \\ \bar{Y}} & - \mathbf{A} 
    } 
    \preceq 0 
    \ \ \forall Q \in \bar{\mathcal{Q}},
    \label{eq:petersenLMIGxu}
\end{align}
then \eqref{eq:LMI_SufficientLyaCondition_data} holds for some $\Gamma \succ 0$, $\epsilon_{\Gamma} > 0$ and $Y$.
\end{lemma}

\begin{proof}
By Assumption~\ref{ass:DataFullRow} and the Schur complement, \eqref{eq:petersenLMIGxu} is rewritten into a form amenable to Petersen's lemma.
After the scaling, Petersen's lemma together with \eqref{eq:DataConsistentSetC_re} yields \eqref{eq:LMI_SufficientLyaCondition_data} for all coefficient matrices $[\hat{E}_A \ \hat{E}_B]\in\mathcal{C}$.
See Appendix~\ref{apx:lema:petersen_data} for details.
\end{proof}

The next result, obtained from Lemma~\ref{lema:polytope_data} and Lemma~\ref{lema:petersen_data}, 
is the solution to Problem~\ref{problem:dataBased}.

\begin{theorem}\label{thm:dataDrivenSolution}
Under Assumptions \ref{ass:function_qi}, \ref{ass:DisturbanceEnergy} and \ref{ass:DataFullRow},
given a positive real number $r$ and the set $\mathcal{B}_r$, 
if there exist $\bar{\Gamma} \succ 0$, $\epsilon_{\bar{\Gamma}} > 0$ and $\bar{Y}$ such that \eqref{eq:petersenLMIGxu} holds,
and let $u = Kx$ with $K = \bar{Y} \bar{\Gamma}^{-1}$, 
then there exists a region of attraction $\mathcal{B}_{\bar{r}_0}$ with
\begin{equation}\label{eq:LMI_r0_data_bar}
    \bar{r}_0 = \min \left\{ r, \ \sqrt{ \tfrac{ \lambda_{\max}(\bar{\Gamma}) \lambda_{\min}(\bar{\Gamma}) }{ \lambda^2_{\max}(\bar{\Gamma}) - \epsilon_{\bar{\Gamma}} \lambda_{\min}(\bar{\Gamma}) } } \cdot r \right\},
\end{equation}
and the closed-loop system $x^{+} = \big( A(x) + B(x)K \big) x$  
starting from any $x(0) \in \mathcal{B}_{\bar{r}_0}$ 
will exponentially converge to the origin.
\end{theorem}

\begin{proof}
Under Assumption~\ref{ass:DataFullRow} and by Lemma~\ref{lema:petersen_data},
\eqref{eq:petersenLMIGxu} implies \eqref{eq:LMI_SufficientLyaCondition_data},
which, under Assumptions~\ref{ass:function_qi} and \ref{ass:DisturbanceEnergy},
leads to the result of this theorem by Lemma~\ref{lema:polytope_data}.
Finally, the common positive scaling used in Lemma~\ref{lema:petersen_data} leaves the ROA expression \eqref{eq:LMI_r0_data} invariant, yielding \eqref{eq:LMI_r0_data_bar}.
\end{proof}

\subsection{Robustness Against Disturbances}\label{sec:robustness_data}

As discussed in Section~\ref{sec:robustness} and Theorem~\ref{thm:Robust2Dist}, 
the model-based controller is robust against disturbances.
The same result applies to the data-driven controller, 
which is detailed in the next theorem.

\begin{theorem}\label{thm:dataDrivenRobust2Dist}
Under Assumptions \ref{ass:function_qi}, \ref{ass:DisturbanceEnergy} and \ref{ass:DataFullRow},
given a positive real number $r$ and the set $\mathcal{B}_r$, 
suppose that there exist $\bar{\Gamma} \succ 0$, $\epsilon_{\bar{\Gamma}} > 0$ and $\bar{Y}$ such that \eqref{eq:petersenLMIGxu} holds, 
and let $K = \bar{Y} \bar{\Gamma}^{-1}$ for the disturbed closed-loop system $x^{+} = \big( A(x) + B(x)K \big) x + w$.
If the initial state $x(0)$ and disturbance $w(k)$ satisfy
\begin{subequations}\label{eq:R2D_condition_data}
\begin{align}
    |x(0)|^2 \leq r^2 \,
    \text{ and } \ \bar{\delta}_{x_0}  |x(0)|^2 + \bar{\delta}_{w}  |w|^2_{\infty} \leq r^2 ,
\end{align}
where
\begin{align}
    \bar{\delta}_{x_0} & := \frac{ 2 \lambda^2_{\max}(\bar{\Gamma}) - \epsilon_{\bar{\Gamma}} \lambda_{\min}(\bar{\Gamma}) } { 2 \lambda_{\min}(\bar{\Gamma}) \lambda_{\max}(\bar{\Gamma}) } ,
    \label{eq:R2D_condition_data_x0}
    \\
    \bar{\delta}_{w} & := \frac{ 4 \gamma^2_{A_{\cl}} \lambda^5_{\max}(\bar{\Gamma}) + 2 \epsilon_{\bar{\Gamma}} \lambda_{\min}(\bar{\Gamma}) \lambda^3_{\max}(\bar{\Gamma}) } { \epsilon^2_{\bar{\Gamma}} \lambda^3_{\min}(\bar{\Gamma}) } ,
    \label{eq:R2D_condition_data_w}
    \\
    \gamma_{A_{\cl}} &:= \max_{x \in \mathcal{B}_r} \| A(x) + B(x)K \| , 
    \label{eq:R2D_condition_data_gamma}
\end{align}
\end{subequations}
then there exists $\bar{\beta} \in \mathcal{KL}$ and $\bar{\gamma} \in \mathcal{K}$ such that
\begin{equation}\label{eq:R2D_result_data}
    |x(k)| \leq \bar{\beta} \big(  |x(0)|  , \, k \big) + \bar{\gamma} (|w|_{\infty}) \quad \forall k \in \mathbb{N}.
\end{equation}
\end{theorem}

\begin{proof}
Under Assumption~\ref{ass:DataFullRow} and by Lemma~\ref{lema:petersen_data},
\eqref{eq:petersenLMIGxu} implies \eqref{eq:LMI_SufficientLyaCondition_data}.
Under Assumptions~\ref{ass:function_qi} and \ref{ass:DisturbanceEnergy},
Lemma~\ref{lema:polytope_data} guarantees the same Lyapunov decrease condition over $\mathcal{B}_r$ as in the model-based case.
Therefore, following the arguments of Theorem~\ref{thm:Robust2Dist} yields the stated robustness result.
Finally, the common positive scaling used in Lemma~\ref{lema:petersen_data} converts \eqref{eq:R2D_condition} into \eqref{eq:R2D_condition_data}.
\end{proof}

\begin{remark}\label{rmk:gammaAcl}
The parameter values in \eqref{eq:R2D_condition}, $\delta_{x_0}$, $\delta_{w}$ and $\gamma_{A_{\cl}}$, are all accessible and computable in the model-based case. 
However, in the data-driven case, the value of $\gamma_{A_{\cl}}$ as defined in \eqref{eq:R2D_condition_data_gamma} is unavailable due to the lack of model knowledge.
Instead, a computable upper bound of $\gamma_{A_{\cl}}$ can be derived and used in its place, that is,
\begin{align*}
    & \gamma_{A_{\cl}} := \max_{x \in \mathcal{B}_r} \| A(x) + B(x)K \|
    \\
    & \overset{\eqref{eq:nSYSrewrite}}{=} 
    \max_{x \in \mathcal{B}_r} 
    \left\lVert [ E_A \ E_B ]  \bmat{ \Xi_A(x) \\ \Xi_B(x) K }  \right\rVert
    \! \leq \! \| [ E_A \ E_B ] \| 
    \max_{x \in \mathcal{B}_r}
    \left\lVert  \bmat{ \Xi_A(x) \\ \Xi_B(x) K  }  \right\rVert
    \\
    & \overset{\eqref{eq:DataConsistentSetC_re}}{\leq}
    \Big( \| \mathbf{Z}_c \| + \| \mathbf{Q}^{\frac{1}{2}} \| \| \mathbf{A}^{-\frac{1}{2}} \| \Big)
    \max_{x \in \mathcal{B}_r}
    \left\lVert \bmat{ \Xi_A(x) \\ \Xi_B(x) K  }  \right\rVert
    \\
    & =: \bar{\gamma}_{A_{\cl}}.
\end{align*}
The upper bound $\bar{\gamma}_{A{\cl}}$, which is computable from data alone, can be used to replace $\gamma_{A_{\cl}}$ in \eqref{eq:R2D_condition_data_gamma}.
\end{remark}

\subsection{System Subject to Input Saturation}\label{sec:InputSatur_data}

This subsection presents a data-driven controller for the system with input saturation, along with an estimate of ROA.

\begin{theorem}\label{thm:LMIsolution_satur_data}
Under Assumptions \ref{ass:function_qi}, \ref{ass:DisturbanceEnergy} and \ref{ass:DataFullRow},
given a positive real number $r$ and the set $\mathcal{B}_r$,
if there exist $\bar{\Gamma} \succ 0$, $\epsilon_{\bar{\Gamma}} > 0$, $\bar{Y}$, $\bar{W}$ and diagonal positive definite matrix $\bar{S}$ such that
\begin{subequations}\label{eq:LMI_SufficientLyaCondition_satur_data}
\begin{align}
    & \bmat{ \bar{\Gamma} & \bar{W}_{(i)}^{\top} \\[3pt] \bar{W}_{(i)} & \bar{u}_i^2} \succeq 0
    \quad i = 1,\dots,m,
    \label{eq:subsetCondi_satur_data}
    \\
    & \bmat{
        -\bar{\Gamma} + \epsilon_{\bar{\Gamma}} I_n
        & \!\! - \bar{Y}^{\top} \!\! - \! \bar{W}^{\top}
        \!\!\! & \!\!\! 0
        \!\!\! & \!\! \smat{\bar{\Gamma} \\ \bar{Y}}^{\!\! \top} \!\! Q^{\top}
        \\[10pt]
        - \bar{Y} - \bar{W} 
        & \!\! -2\bar{S}
        \!\!\! & \!\!\! 0
        \!\!\! & \!\! \smat{ 0 \\ \bar{S}}^{\!\! \top} \!\! Q^{\top}
        \\[10pt]
        0 & \!\! 0 \!\!\! & \!\!\! -\bar{\Gamma} - \mathbf{C} \!\!\! & \!\! - \mathbf{B}
        \\[4pt]
        Q \smat{ \bar{\Gamma} \\ \bar{Y}}
        & \!\! Q \smat{ 0 \\ \bar{S}}
        \!\!\! & \!\!\! - \mathbf{B}^{\top}
        \!\!\! & \!\! - \mathbf{A}
    } \preceq 0 
    \ \forall Q \! \in \! \bar{\mathcal{Q}}
    \label{eq:stableCondi_satur_data}
\end{align}
\end{subequations}
and let $u = Kx$ with $K = \bar{Y} \bar{\Gamma}^{-1}$, 
then there exists a region of attraction 
$\mathcal{B}_{\bar{r}_0} \cap \mathcal{E}_{\bar{P},1}$
with $\bar{P} = \bar{\Gamma}^{-1}$ and
\begin{equation}\label{eq:modelBased_r0_satur_thm_data}
    \bar{r}_0 = \min \left\{ r, \ \sqrt{ 
        \tfrac{ \lambda_{\max}(\bar{\Gamma}) \lambda_{\min}(\bar{\Gamma}) }
        { \lambda^2_{\max}(\bar{\Gamma}) - \epsilon_{\bar{\Gamma}} \lambda_{\min}(\bar{\Gamma}) }
    } \cdot r \right\} , 
\end{equation}
and the closed-loop system $x^{+} = A(x)x + B(x)\sat(Kx)$ 
starting from any $x(0) \in \mathcal{B}_{\bar{r}_0} \cap \mathcal{E}_{\bar{P},1}$ 
will exponentially converge to the origin.
\end{theorem}

\begin{proof}
Condition~\eqref{eq:subsetCondi_satur_data} guarantees the sector condition of Lemma~\ref{lema:sectorCondi}, while \eqref{eq:stableCondi_satur_data}, through Petersen's lemma and the polytopic inclusion \eqref{eq:Q_subset_convQ}, yields the corresponding model-based saturation condition over $\mathcal{B}_r$ for the true system.
Combining the two gives the Lyapunov condition \eqref{eq:LyaCondition_satur}, and Lemma~\ref{lema:intermeLyaCondi} then yields the stated ROA.
See Appendix~\ref{apx:thm:LMIsolution_satur_data} for details.
\end{proof}

\subsection{Discussion on Region of Attraction}\label{sec:roa}

\subsubsection{ROA-Oriented Controller Synthesis}
\leavevmode\par

The expression of $\bar r_0$ in \eqref{eq:LMI_r0_data_bar} suggests how to construct an objective that favors a larger ROA estimation.  
From \eqref{eq:LMI_r0_data_bar}, let $\lambda_-:=\lambda_{\min}(\bar\Gamma)$, $\lambda_+:=\lambda_{\max}(\bar\Gamma)$, and define the ROA-related factor 
$\phi_{\rm ROA}:=\lambda_+\lambda_-/(\lambda_+^2-\epsilon_{\bar\Gamma}\lambda_-)$.
When 
$\lambda_+^2-\epsilon_{\bar\Gamma}\lambda_->0$, 
we have 
$\partial\phi_{\rm ROA}/\partial\epsilon_{\bar\Gamma}=\lambda_+\lambda_-^2/(\lambda_+^2-\epsilon_{\bar\Gamma}\lambda_-)^2>0$, 
$\partial\phi_{\rm ROA}/\partial\lambda_-=\lambda_+^3/(\lambda_+^2-\epsilon_{\bar\Gamma}\lambda_-)^2>0$, 
and 
$\partial\phi_{\rm ROA}/\partial\lambda_+=-\lambda_-(\lambda_+^2+\epsilon_{\bar\Gamma}\lambda_-)/(\lambda_+^2-\epsilon_{\bar\Gamma}\lambda_-)^2<0$. 
A larger certified ROA is thus favored by a smaller $\lambda_+$, a larger $\lambda_-$, and a larger $\epsilon_{\bar\Gamma}$. 
A convex surrogate can be obtained by introducing
$\underline{\lambda}_{\bar\Gamma}$ and
$\bar{\lambda}_{\bar\Gamma}$ satisfying
$\underline{\lambda}_{\bar\Gamma} I_n \preceq \bar\Gamma \preceq \bar{\lambda}_{\bar\Gamma} I_n$,
and augmenting the LMIs in Theorem~\ref{thm:dataDrivenSolution} with the objective
\vspace{-1mm}
\begin{equation}\label{eq:objectiveROA}
    \underset{}{\operatorname{minimize}}\quad
    c_{\max}\bar{\lambda}_{\bar\Gamma}
    -c_{\min}\underline{\lambda}_{\bar\Gamma}
    -c_{\epsilon}\epsilon_{\bar\Gamma}
    \vspace{-1mm}
\end{equation}
with $c_{\max}, c_{\min}, c_{\epsilon} \geq 0$.
This convex surrogate tends to increase the estimated ROA, although it does not exactly maximize $\bar r_0$ and direct maximization of $\bar r_0$ does not preserve a convex SDP formulation. The additional objective does not change the original LMI conditions, so the stability guarantee in Theorem~\ref{thm:dataDrivenSolution} is preserved.
The same analysis applies to the model-based case in Section~\ref{sec:ModelBased}, by replacing $\bar\Gamma$ and $\epsilon_{\bar\Gamma}$ with $\Gamma$ and $\epsilon_{\Gamma}$, respectively.

The choice of $r$ affects both the conservatism of the polytopic over-approximation and the size of the certified ROA.
A smaller $r$ tightens the element-wise bounds and hence reduces the conservatism in \eqref{eq:G_subset_convG} and \eqref{eq:Q_subset_convQ}, but at the same time restricts the region over which the ROA can be certified in \eqref{eq:LMI_r0} and \eqref{eq:LMI_r0_data_bar}.
In future work, $r$ could be iteratively tuned offline, together with the controller synthesis, to seek a larger ROA.

\subsubsection{Complementary ROA Estimation} 
\leavevmode\par

While the ROA is directly estimated as $\mathcal{B}_{r_0}$ in Theorem~\ref{thm:LMIsolution}, this is not the only method.
In the model-based case, the closed-loop system $x^{+} = A(x)x + B(x)Kx$ is known, and the Lyapunov function is given by $V(x) = x^{\top} \Gamma^{-1} x$, where $\Gamma$ is computed from \eqref{eq:LMI_SufficientLyaCondition}. Then $V(x^{+}) - V(x) = v(x)$ with
\vspace{-1mm}
\begin{align*}
    v(x)
    := & \big[ \ast \big]^{\top} \Gamma^{-1} \cdot \big[ A(x)x + B(x)Kx \big] 
    - x^{\top} \Gamma^{-1} x ,
\end{align*}
and we can numerically determine the set $\mathcal{V} := \{ x : v(x) < 0 \}$.
Any sublevel set 
$\mathcal{E}_{\Gamma^{-1}, \gamma} := \{ x \in \real^{n} : x^{\top} \Gamma^{-1} x \leq \gamma \}$
contained in $\mathcal{V} \cup \{0\}$ is an estimate of ROA for the closed-loop system.
Moreover, the union of the set $\mathcal{B}_{r_0}$ and the largest sublevel set $\mathcal{E}_{\Gamma^{-1}, \gamma}$ contained in $\mathcal{V} \cup \{0\}$, i.e., $\mathcal{B}_{r_0} \cup \mathcal{E}_{\Gamma^{-1}, \gamma}$, is also an ROA.

When the system is subject to input saturation, the same procedure applies.
The closed-loop system is $x^{+} = A(x)x + B(x) \sat(Kx)$, and the Lyapunov function is given by $V(x) = x^{\top} \Gamma^{-1} x$, where $\Gamma$ is computed from \eqref{eq:LMI_SufficientLyaCondition_satur}. Then, $V(x^{+}) - V(x) = v_{\sat}(x)$ with
\begin{align*}
    v_{\sat}(x) := \big[ \ast \big]^{\top} \Gamma^{-1} \cdot \big[ A(x)x + B(x)\sat(Kx) \big]
    - x^{\top} \Gamma^{-1} x ,
\end{align*}
and we can numerically determine the set $\mathcal{V}_{\sat} := \{ x : v_{\sat}(x) < 0 \}$.
Any sublevel set $\mathcal{E}_{\Gamma^{-1}, \gamma}$
contained in $\mathcal{V}_{\sat} \cup \{0\}$ is an estimate of ROA for the closed-loop system.
The union of the set $\mathcal{B}_{r_0} \cap \mathcal{E}_{\Gamma^{-1},1}$, estimated by Theorem~\ref{thm:LMIsolution_satur}, and the largest sublevel set $\mathcal{E}_{\Gamma^{-1}, \gamma}$ contained in $\mathcal{V}_{\sat} \cup \{0\}$, i.e., $\{ \mathcal{B}_{r_0} \cap \mathcal{E}_{\Gamma^{-1},1} \} \cup \mathcal{E}_{\Gamma^{-1}, \gamma}$, is also an ROA.

In the data-driven case, $v(x)$ and $\mathcal{V}$ are unavailable in the absence of model knowledge, yet an upper bound on $v(x)$ can still be computed from the data.
By Assumption~\ref{ass:function_qi} and \eqref{eq:nSYSrewrite}, the closed-loop system is $x^{+} = E_A \Xi_A(x)x + E_B \Xi_B(x) K x$.
Assumption~\ref{ass:DisturbanceEnergy} ensures $[ E_A \ E_B ] \in \mathcal{C}$.
Assumption~\ref{ass:DataFullRow} allows us to express the set $\mathcal{C}$ in the form of \eqref{eq:DataConsistentSetC_re}, which means 
$[ E_A \ E_B ] = \mathbf{Z}_c + \mathbf{Q}^{1/2} \Upsilon \mathbf{A}^{-1/2}$ 
for some $\Upsilon$ with $\| \Upsilon \| \leq 1$,
where $\mathbf{Z}_c$, $\mathbf{Q}$ and $\mathbf{A}$ are defined by data. 
The Lyapunov function is given by $V(x) = x^{\top} \Gamma^{-1} x$. 
By the scaling relation $\Gamma^{-1}=\tau\bar{\Gamma}^{-1}$ for some $\tau>0$, with $\bar{\Gamma}$ obtained from \eqref{eq:petersenLMIGxu}, we have
\vspace{-1mm}
\begin{align*}
    V(x^{+}) - V(x) 
    & = \big[ \ast \big]^{\top} \Gamma^{-1} \cdot [ E_A \ E_B ] \bmat{ \Xi_A(x)x \\ \Xi_B(x) K x } 
    - x^{\top} \Gamma^{-1} x
    \\ 
    & \leq \tau \Big(  
       \underbrace{ - x^{\top} \bar{\Gamma}^{-1} x + v_1(x) + v_2(x) + v_3(x)}_{=: v_{\rm{data}}(x)}
    \Big),
\end{align*}
where
\begin{align*}
    v_1(x) & := \bmat{ \Xi_A(x)x \\ \Xi_B(x) K x }^{\top} \mathbf{Z}_c^{\top} \bar{\Gamma}^{-1} \mathbf{Z}_c \bmat{ \Xi_A(x)x \\ \Xi_B(x) K x } ,
    \\
    v_2(x) & := 2 \left\lvert \mathbf{Q}^{\frac{1}{2}} \bar{\Gamma}^{-1} \mathbf{Z}_c \!\! \bmat{ \Xi_A(x)x \\ \Xi_B(x) K x } \right\rvert \left\lvert \mathbf{A}^{-\frac{1}{2}} \!\! \bmat{ \Xi_A(x)x \\ \Xi_B(x) K x } \right\rvert  ,
    \\
    v_3(x) & := \left\| \mathbf{Q}^{\frac{1}{2}} \bar{\Gamma}^{-1} \mathbf{Q}^{\frac{1}{2}} \right\| \left\lvert \mathbf{A}^{-\frac{1}{2}} \!\! \bmat{ \Xi_A(x)x \\ \Xi_B(x) K x } \right\rvert^2 .
\end{align*}
The value of $v_{\rm{data}}(x)$ is computable from data, thus we can numerically determine the set $\mathcal{V}_{\rm{data}} := \{ x : v_{\rm{data}}(x) < 0 \}$.
Any sublevel set $\mathcal{E}_{\bar{\Gamma}^{-1}, \gamma}$ contained in $\mathcal{V}_{\rm{data}} \cup \{0\}$ is an estimate of ROA for the closed-loop system.
Moreover, the union of the set $\mathcal{B}_{\bar{r}_0}$ and the largest sublevel set $\mathcal{E}_{\bar{\Gamma}^{-1}, \gamma}$ contained in $\mathcal{V}_{\rm{data}} \cup \{0\}$, i.e., $\mathcal{B}_{\bar{r}_0} \cup \mathcal{E}_{\bar{\Gamma}^{-1}, \gamma}$, is also an ROA.

When the system is subject to input saturation, the same procedure applies in the data-driven case, and we omit the details due to space limitations.

\section{Comparative Studies and Physical Experiment}\label{sec:experiments}

\subsection{Comparative Studies}\label{sec:comparative}

To assess computational scalability, we compare Theorem~\ref{thm:dataDrivenSolution} with the SOS baseline in~\cite[Cor.~1]{Guo2022SOS} on the family of polynomial systems
$\dot{x} = ( 0.1 I_n + A_{\mathrm{sh},n} ) x - \frac{1}{\sqrt{n}} B_{e_n} \sum_{i=1}^{n}x_i^3 + B_{e_n}u$,
where $A_{\mathrm{sh},n}$ is the nilpotent shift matrix with ones on the superdiagonal, and $B_{e_n}=[0~\cdots~0~1]^\top$.
We discretize the above polynomial systems using the forward Euler method with step size $h = 0.01$ and apply Theorem~\ref{thm:dataDrivenSolution} with $r=0.3$ to the resulting discrete-time systems.
For the SOS baseline, we choose $Z(x) = ( x^\top,x_1^3,\ldots,x_n^3 )^\top$ and $\widehat Z(x)=x$.
For Theorem~\ref{thm:dataDrivenSolution}, we choose $\xi_{A_j}(x)=(1,x_j^2)$, $j=1,\ldots,n$.
For each $n$, 30 independent trials are conducted. 
Each trial comprises 50 trajectories, each excited by an independently generated input sequence of length 10. The resulting $50\times 10=500$ data samples are aggregated into a single dataset.
The continuous-time system is excited by $u(t)=u_k$ for $t\in[kh,(k+1)h)$, and the same $u_k$ is applied to the discrete-time system at step $k$. Thus, for each $n$, we obtain 30 continuous-time datasets for the SOS baseline and 30 paired discrete-time datasets for Theorem~\ref{thm:dataDrivenSolution}: within a pair the two trajectories share the same initial condition, and the same $30$ input excitations are reused for every $n$.
No noise is injected, so the comparison focuses exclusively on computational scalability.
All trials are sequentially run on an Intel Core i5-13500 desktop with 16~GB of RAM, using MATLAB R2025b, MOSEK, and YALMIP~\cite{YALMIP}.
A trial counts as failed if the solve time exceeds 500\hspace{0.1em}s or the solver runs out of memory; mean solve times are computed over successful trials only.
Theorem~\ref{thm:dataDrivenSolution} yields a linear feedback $u=Kx$, so we compare it against the SOS baseline with a constant $K(x)$, which is zero order and yields linear feedback; orders one through four are also tested to examine how the order affects feasibility and solve time.
Table~\ref{tab:scalability_comparison} demonstrates a clear scalability advantage of our method: the SOS baseline deteriorates sharply beyond $n=3$ and yields no successful solutions for the tested cases with $n\geq6$, whereas our method remains highly reliable up to $n=12$ with substantially lower solve times.

\begin{table}[t]
    \centering
    \begin{threeparttable}
    \caption{Scalability comparison.}
    \label{tab:scalability_comparison}

    \scriptsize
    \renewcommand{\arraystretch}{1.08}
    \setlength{\tabcolsep}{2.5pt}

    \begin{tabular*}{\columnwidth}{@{\extracolsep{\fill}}ccccccc@{}}
        \toprule
        \multirow{2}{*}{\textbf{$n$}}
        & \multirow{2}{*}{Ours}
        & \multicolumn{5}{c}{SOS baseline: order of $K(x)$} \\
        \cmidrule(lr){3-7}
        & & {0} & {1} & {2} & {3} & {4} \\
        \midrule
         2 & 100 / .078 &  100 / .113 &  100 / .111 &  100 / .115 & 100 / .135 & 100 / .135 \\
         3 & 100 / .083 & 96.7 / .154 & 96.7 / .164 & 96.7 / .161 & 100 / .727 & 100 / .752 \\
         4 & 100 / .097 & 66.7 / .332 & 66.7 / .355 &  6.7 / .387 & 96.7 / 13.6 & 100 / 13.3 \\
         5 & 100 / .137 &    0 / NV   &    0 / NV   &    0 / NV   & 96.7 / 238 & 83.3 / 244 \\
         6 & 100 / .247 &    0 / NV   &    0 / NV   &    0 / NV   &    0 / OM  &   0 / OM   \\

         $\vdotsS$ & & & & & & \\

        11 & 100 / 33.9 & \multicolumn{5}{c}{--} \\
        12 & 96.7 / 106 & \multicolumn{5}{c}{--} \\
        13 &   0 / TO   & \multicolumn{5}{c}{--} \\
        \bottomrule
    \end{tabular*}

    \begin{tablenotes}[flushleft]
        \scriptsize
        \item[] Each entry reports the successful rate (\%) / mean solve time ($\mathrm{s}$).
        TO, OM, and NV respectively denote timeout ($>500~\mathrm{s}$), out of memory, and a solver-returned solution that does not yield a valid stabilizing controller.
        ``--'' indicates that the corresponding case is not solved further.
    \end{tablenotes}
    \end{threeparttable}
    \vspace{-4mm}
\end{table}

Next, we compare our ROA estimation with the method in \cite{Claudio2023Cancel}.
This allows a fair comparison using the same system and noisy dataset.
Consider the nonlinear system
\vspace{-2mm}
\begin{equation}\label{eq:example_sys}
\begin{cases}
    x_{1}^{+} = x_1 + 0.1 \sin{x_1} + 0.2 x_2
    + (0.1 + 0.1 |x_2|) u, \\
    x_{2}^{+} = 0.2 x_1 + 0.9 x_2 + 0.1 x_1^2 x_2
    + 0.1 e^{x_1} u .
\end{cases}
\vspace{-2mm}
\end{equation}
Let
$\xi_{A_1}(x)=(1,\sinc x_1)$,
$\xi_{A_2}(x)=(1,x_1,x_1^2)$, and
$\xi_{B_1}(x)= (1,|x_1|,|x_2|,e^{x_1},e^{x_2})$,
where $\sinc x_1:=\sin x_1/x_1$ for $x_1\neq0$ and $\sinc(0):=1$.
Assumption~\ref{ass:function_qi} allows us to write the state-dependent representation as $A(x) = E_A \blkdiag( \xi_{A_1}(x), \xi_{A_2}(x) )$ and $B(x) = E_B \xi_{B_1}(x)$, where $E_A$ and $E_B$ are unknown in the controller design.
As in \eqref{eq:nSYSrewrite_distur}, 10 experiments are performed, each consisting of 13 samples, with the input $u_d$ uniformly distributed over $[-1.3,1.3]$ and the random noise satisfying the energy bound $\Theta=0.0021I_2$.
We store the data samples into $X_0 \in \real^{5 \times 130}$, $X_1 \in \real^{2 \times 130}$, $U_0 \in \real^{5 \times 130}$ as in \eqref{eq:dataMatrix}.
Take $r = 0.92$ for $\mathcal{B}_r$. 
By \eqref{eq:min_max_qji}, we have the set $\bar{\mathcal{Q}}$ as in \eqref{eq:bar_set_Q}.
Solving \eqref{eq:petersenLMIGxu} in Theorem~\ref{thm:dataDrivenSolution} via YALMIP yields  
$\epsilon_{\bar \Gamma} = 0.000045$, $\bar Y = [-0.0178 \ -0.0186]$, 
$\bar \Gamma = \smat{0.0082 & 0.0000 \\ 0.0000 & 0.0082}$, and 
$K =  [-2.1797 \ -2.2750]$.
According to \eqref{eq:LMI_r0_data_bar}, 
the estimate of ROA, $\mathcal{B}_{\bar{r}_0}$, has a radius of $\bar{r}_0 = 0.92$.
We then apply the method of \cite{Claudio2023Cancel} to the same dataset.
As the dynamics of \eqref{eq:example_sys} are more general than \cite[Eq.~(1)]{Claudio2023Cancel}, we resort to \cite[Cor.~1]{Claudio2023Cancel}. However, this result assumes noise-free data, so we adapt \cite[Lem.~3]{Claudio2023Cancel} to make \cite[Cor.~1]{Claudio2023Cancel} applicable to noisy settings.
For \cite[Eq.~(34)]{Claudio2023Cancel}, we set 
$\mathcal{Z}(x,u) = (x,\ u,\ \sin x_1 - x_1,\ x_1^2 x_2,\ |x_2| u,\ e^{x_1} u - u)$.
Take $\Omega = I_3$ in \cite[Eq.~(48)]{Claudio2023Cancel}.
Solving the data-based LMIs in \cite{Claudio2023Cancel} yields the controller
$u^{+} = -2.5135 x_1 -1.4326 x_2 + 0.2220 u$.
The resulting Lyapunov function is $V(\chi) = \chi^{\top} P_1^{-1} \chi$ with $\chi := (x, u)$ and 
$P_1^{-1} = 
\smat{
    0.0261  &  0.0040  &  0.0048 \\
    0.0040  &  0.0154  &  0.0026 \\
    0.0048  &  0.0026  &  0.0020
}$.
The ROA is estimated using the adapted version of \cite[Prop.~2]{Claudio2023Cancel}, and is depicted in Figure~\ref{fig:data_comparisom}. 
Specifically, the gray area in Figure~\ref{fig:data_comparisom} represents the set $\mathcal{L}$ from \cite[Prop.~2]{Claudio2023Cancel}, where $V(\chi^{+}) - V(\chi)$ is negative. 
The black dotted ellipsoid is the Lyapunov sublevel set $\mathcal{R}_{\gamma} := \{ \chi : V(\chi) \leq \gamma \}$ contained in $\mathcal{L}\cup\{0\}$, hence, it provides an estimate of ROA. 
Both sets $\mathcal{L}$ and $\mathcal{R}_{\gamma}$ are projected onto the plane $\{ \chi : \chi_3 = 0 \}$.
The red solid circle in Figure~\ref{fig:data_comparisom} represents the ROA estimated by our method, which is larger than that obtained from \cite{Claudio2023Cancel}.
We note that the calculation of $\mathcal{L}$ in \cite[Prop.~2]{Claudio2023Cancel} is based on an upper bound of $V(\chi^{+}) - V(\chi)$, which introduces additional conservatism. This may explain why our approach yields a larger estimate of ROA.

\vspace{-2mm}
\subsection{Physical Experiment}\label{sec:physical}

The data-driven controller is validated on a physical platform using the ROS--PX4 framework~\cite{XTDrone}.
The quadrotor is mounted on a three-degree-of-freedom rotational platform and follows the attitude dynamics
\vspace{-2mm}
\begin{equation*}
\begin{cases}
    \dot{\eta}=T(\phi,\theta)\omega,\\
    J\dot{\omega}+\omega\times(J\omega)=M,\\
    M=\blkdiag(c_1,c_2,c_3)u,
\end{cases}
\vspace{-2mm}
\end{equation*}
where $J=\blkdiag(J_{11},J_{22},J_{33})\in\mathbb{R}^{3\times3}$ is the inertia matrix with respect to the body-fixed frame; $\omega\in\mathbb{R}^3$ and $M\in\mathbb{R}^3$ are the body-frame angular velocity and total moment vector, respectively; $u=(u_1,u_2,u_3)\in\mathbb{R}^3$ is the PX4 control input; $\times$ denotes the vector cross product; $\eta=(\phi,\theta,\psi)\in\mathbb{R}^3$ denotes the Euler angles (roll, pitch, and yaw) describing the attitude with respect to the inertial frame; and $c_i$, $i=1,2,3$, is the proportionality coefficient relating $u_i$ to the corresponding moment component $M_i$. The transformation matrix $T(\phi,\theta)$ is given in \cite[Eq.~(5.67)]{wie1998space}. The parameters $J$ and $c_i$, $i=1,2,3$, are unknown and are not required for the data-driven controller synthesis.

\begin{figure}[t]
    \centering
    \includegraphics[width=0.60\linewidth]{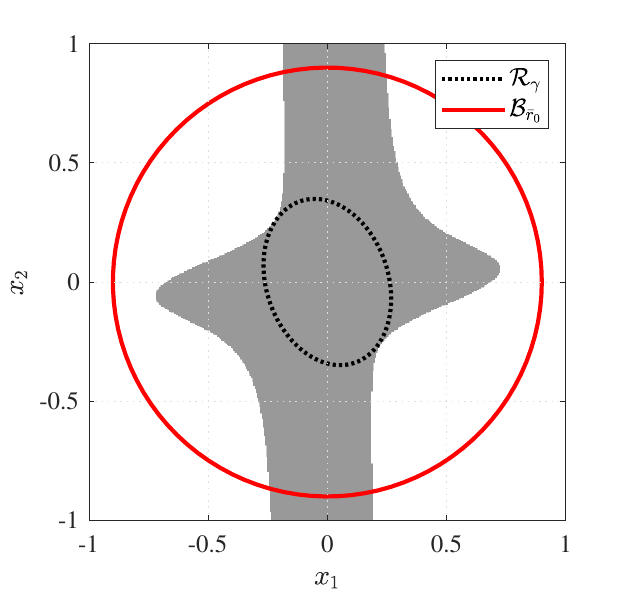}     
    \vspace{-3mm}
    \caption{Comparison results. The black dotted ellipsoid $\mathcal{R}_{\gamma}$ is the ROA estimated by \cite{Claudio2023Cancel}; The red sold circle $\mathcal{B}_{\bar{r}_0}$ is the ROA estimated by Theorem~\ref{thm:dataDrivenSolution}.}
    \label{fig:data_comparisom}
    \vspace{-3mm}
\end{figure}

We discretize the continuous-time system using the forward Euler method with a sampling time $h=0.002$. The resulting discrete-time system can be written in the form of \eqref{eq:nSYS} with $x=(\eta,\omega)$. 
For Assumption~\ref{ass:function_qi}, we have the known basis functions:
\begin{align*}
    & 
    \xi_{B_i}(x)=1,\ i=1,2,3, \
    \xi_{A_i}(x)=1,\ i=1,2,3, \
    \xi_{A_4}(x)=(1,x_6),
    \\
    &
    \xi_{A_5}(x)
    =(
    1, \ \sin x_1\tan x_2, \ \cos x_1, \ 
    \tfrac{\sin x_1}{\cos x_2}, \ x_4 ),
    \\
    &
    \xi_{A_6}(x)
    =(
    1, \ \cos x_1\tan x_2, \  \sin x_1, \ 
    \tfrac{\cos x_1}{\cos x_2}, \ x_5 ).
\end{align*}

Details of the experimental data collection are provided in Appendix~\ref{apx:ExpDetails}.
Store the data samples into $X_0$, $X_1$ and $U_0 $ as in \eqref{eq:dataMatrix}.
The process noise $D_0$ as in \eqref{eq:dataMatrix_D0} is bounded by $\Theta = 0.01I_6$ as in \eqref{eq:DisturbanceEnergy}. 
The derivation of the noise bound is lengthy and is deferred to Appendix~\ref{apx:ExpDetails}.
Take $r = 0.5$ for $\mathcal{B}_r$.
By \eqref{eq:min_max_qji}, we have the set $\bar{\mathcal{Q}}$ as in \eqref{eq:bar_set_Q}, which contains $2^9=512$ distinct elements.
To favor a larger ROA and a nearly diagonal controller structure, we augment the feasibility problem in Theorem~\ref{thm:dataDrivenSolution} with the ROA-oriented objective in \eqref{eq:objectiveROA}, as discussed in Section~\ref{sec:roa}, together with a penalty on the off-diagonal entries of $\bar{Y}$. Solving this SDP via YALMIP yields
\vspace{-2mm}
\[
    K = \smat{
   -0.6528&    0.0000&    0.0000&   -0.0398&    0.0000&    0.0000\\
    0.0000&    1.2725&    0.0000&    0.0000&    0.1292&    0.0000\\
    0.0000&    0.0000&    0.7538&    0.0000&    0.0000&    0.1394
    }.
    \vspace{-1mm}
\]
According to \eqref{eq:LMI_r0_data_bar}, an estimate of ROA, $\mathcal{B}_{\bar{r}_0}$, has a radius of $\bar{r}_0 = 0.5$.
The SDP is solved in $6.0\,\mathrm{s}$ using the computer in Section~\ref{sec:comparative}.
The state response of the closed-loop system is illustrated in Figure~\ref{fig:quadrotor_state_response}. 

Additional demonstrations, which are not included here due to space limitations, are available at our project website \cite{github:AxxBxu} for interested readers.

\begin{figure}[t]
    \centering
    \includegraphics[width=0.85\linewidth]{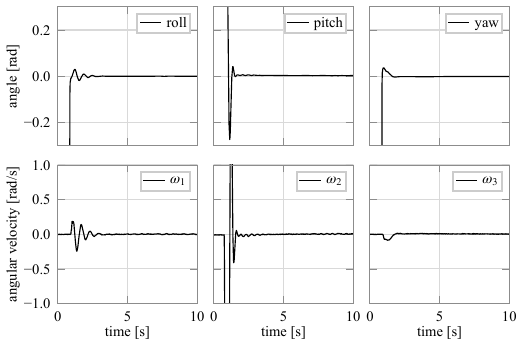}   
    \vspace{-3mm}
    \caption{State response of the closed-loop system.}
    \label{fig:quadrotor_state_response}
    \vspace{-4mm}
\end{figure}

\section{Conclusion}\label{sec:conclusion}

In this paper, we have synthesized a controller for a nonlinear system directly from noisy data with rigorous end-to-end guarantees.
Specifically, the state-dependent representation of a nonlinear system was utilized to derive model-based conditions for controller design.
These conditions were subsequently adapted to the data-driven case.
The resulting direct data-driven controller, with a guaranteed region of attraction, can stabilize all data-consistent nonlinear systems induced by data noise. 
Moreover, the closed-loop properties, namely robustness against disturbances and region of attraction under input saturation, were rigorously characterized based solely on data.

Although our method shows better scalability than the SOS baseline in the considered benchmark, the number of polytope vertices can grow exponentially with the system dimension in the worst case, increasing the computational burden for higher-dimensional systems. Exploiting the state-induced correlations among the entries of $(A(x),B(x))$ to reduce this burden is an important direction for future work. In addition, basis-function mismatch can be treated as process noise during data collection and as an additive disturbance during controller implementation within our current framework. A more refined formulation that explicitly distinguishes basis-function mismatch from data noise and external disturbances is also of interest. Another direction is the data-driven design of state-dependent feedback gains \cite{Cimen2010SDRE} or multiple feedback gains as in parallel distributed compensation \cite{Wang1996PDC}, which may reduce conservatism and enlarge the certified ROA compared with fixed-gain controllers. However, extending these controller structures to the direct data-driven setting is nontrivial and requires new theoretical developments.

\bibliographystyle{IEEEtran}

\bibliography{mybib}

\appendices

\section{Proofs of the Theoretical Results}

\subsection{Proof of Lemma~\ref{lema:modelBasedSolution}}
\label{apx:lema:modelBasedSolution}

\begin{proof}[Proof of Lemma~\ref{lema:modelBasedSolution}]
Let $V(x) := x^{\top} P x$ and $P \succ 0$. 
For all $x \in \mathcal{B}_{r}$, 
\begin{align*}
    & V(x^{+}) - V(x) 
    \\
    & = x^{\top} \Big( [A(x) + B(x)K]^{\top} P [A(x) + B(x)K] - P \Big) x 
    \\
    & \overset{\eqref{eq:LyaCondition}}{\leq} - \epsilon \, |x|^{2} 
    \leq - \frac{\epsilon}{\lambda_{\max}(P)} V(x) .
\end{align*} 
The above inequalities show that \eqref{eq:LyaCondition} implies
\begin{align}\label{eq:LyaDecrease}
    V(x(k+1)) \leq \mu V(x(k)) \ \text{ with } \ \mu := 1 - \frac{\epsilon}{\lambda_{\max}(P)}
\end{align} 
for all $x \in \mathcal{B}_{r}$.
Moreover, \eqref{eq:LyaCondition1} implies $P-\epsilon I_n \succeq 0$ and hence $\epsilon \leq \lambda_{\min}(P)$. Since \eqref{eq:LyaCondition1} remains valid for any smaller positive $\epsilon$, without loss of generality, we take $0 < \epsilon < \lambda_{\min}(P) \leq \lambda_{\max}(P)$. Therefore, $\lambda_{\max}(P)-\epsilon>0$, the denominator in \eqref{eq:modelBased_r0} is positive, and $0<\mu<1$.
Next we claim that, given \eqref{eq:LyaDecrease}, if $x(0) \in \mathcal{B}_{r_0}$, 
then $x(k) \in \mathcal{B}_{r}$ for all $k \geq 1$.
This claim is proven in two cases. \\
Case (i): $\sqrt{ {\lambda_{\min}(P)} / \left( {\lambda_{\max}(P)-\epsilon} \right) } \leq 1$. 
Then $r_0 \leq r$, and $x(0) \in \mathcal{B}_{r_0} \subseteq \mathcal{B}_r$.
According to \eqref{eq:LyaDecrease}, $V(x(1)) \leq \mu V(x(0))$, which implies
\begin{align*}
    |x(1)| & \leq \sqrt{ \frac{\mu \lambda_{\max}(P)}{\lambda_{\min}(P)} } |x(0)| 
    \leq \sqrt{ \frac{\mu \lambda_{\max}(P)}{\lambda_{\min}(P)} } r_0 
    \\
    & = \sqrt{ \frac{\mu \lambda_{\max}(P)}{\lambda_{\min}(P)} } 
    \sqrt{ \frac{\lambda_{\min}(P)} {\lambda_{\max}(P)-\epsilon} } 
    \cdot r 
    = r ,
\end{align*}
and means $x(1) \in \mathcal{B}_r$.
Let $x(i) \in \mathcal{B}_{r}$ for $i = 1,\dots,k$; 
this condition together with \eqref{eq:LyaDecrease} implies
\begin{align*}
    & V(x(k+1)) \leq \mu^{k+1} V(x(0)) 
    \\
    \implies & |x(k+1)| 
    \leq \mu^{\frac{k}{2}} \sqrt{ \frac{\mu \lambda_{\max}(P)}{\lambda_{\min}(P)} } |x(0)| 
    \leq \mu^{\frac{k}{2}} \cdot r < r 
    \\
    \implies & x(k+1) \in \mathcal{B}_r .
\end{align*}
By induction, we conclude that $x(k) \in \mathcal{B}_{r}$ for all $k \geq 1$.  \\
Case (ii): $\sqrt{ {\lambda_{\min}(P)} / \left( {\lambda_{\max}(P)-\epsilon} \right) } > 1$. 
Then $r_0 = r$, and $x(0) \in \mathcal{B}_{r_0} = \mathcal{B}_r$.
According to \eqref{eq:LyaDecrease}, $V(x(1)) \leq \mu V(x(0))$, which implies
\begin{align*}
    |x(1)| \leq \sqrt{ \frac{\mu \lambda_{\max}(P)}{\lambda_{\min}(P)} } |x(0)| 
    \leq \sqrt{ \frac{\mu \lambda_{\max}(P)}{\lambda_{\min}(P)} } \cdot r < r ,
\end{align*}
and means $x(1) \in \mathcal{B}_r$.
Let $x(i) \in \mathcal{B}_{r}$ for $i = 1,\dots,k$; 
this condition together with \eqref{eq:LyaDecrease} implies $x(k+1) \in \mathcal{B}_r$ by analogous arguments as in Case (i).
By induction, $x(k) \in \mathcal{B}_{r}$ for all $k \geq 1$. \\
So far we conclude that, under condition \eqref{eq:LyaCondition}, 
$x(k) \in \mathcal{B}_{r}$ for all $k \geq 0$ if $x(0) \in \mathcal{B}_{r_0}$.
This fact implies, for all $k \in \mathbb{N}$,
\begin{align*}
    & V(x(k+1)) \leq \mu V(x(k)) \implies V(x(k)) \leq \mu^{k} V(x(0))  
    \\
    & \implies |x(k)| \leq \sqrt{ \lambda_{\max}(P) / \lambda_{\min}(P) } \ \mu^{{k}/{2}}  |x(0)|,
\end{align*} 
which completes the proof.
\end{proof}

\subsection{Proof of Theorem~\ref{thm:LMIsolution}}
\label{apx:thm:LMIsolution}

\begin{proof}[Proof of Theorem~\ref{thm:LMIsolution}]
For fixed $\Gamma$, $Y$, and $\epsilon_\Gamma$, the matrix-valued function
\begin{equation*}
    c(G) := 
    \bmat{ 
        -\Gamma & G \bmat{ \Gamma \\ Y} \\
        [\Gamma \ Y^{\top} ] G^{\top} & - \Gamma + \epsilon_{\Gamma} I_n
    }
\end{equation*}
is affine with respect to $G$.
By affinity, convex combinations of $G$ are mapped to the corresponding convex combinations of $c(G)$.
Then by \cite[Prop. 1.14]{LMI2000Scherer}, \eqref{eq:LMI_SufficientLyaCondition} is equivalent to
\begin{align*}
    c(G) \preceq 0 \ \ \forall G \in \bar{\mathcal{G}} 
    \iff 
    c(G) \preceq 0 \ \ \forall G \in {\rm{conv}}(\bar{\mathcal{G}}) ,
\end{align*}
which implies $c(G) \preceq 0$ for all $G \in \mathcal{G}$ by \eqref{eq:G_subset_convG}.
In other words, recalling the definition of $\mathcal{G}$ in \eqref{eq:GxOO}, 
the linear matrix inequalities \eqref{eq:LMI_SufficientLyaCondition} implies
\begin{equation}\label{eq:thm_model_step1}
    \bmat{ 
        -\Gamma 
        & \hspace{-5mm} A(x) \Gamma + B(x) Y 
        \\
        \Gamma A(x)^{\top} + Y^{\top} B(x)^{\top} 
        & \hspace{-5mm} - \Gamma + \epsilon_{\Gamma} I_n
    } \preceq 0
    \ \ \forall x \in \mathcal{B}_r.
\end{equation}
Provided $\Gamma \succ 0$, $Y = K \Gamma$, by Schur complement, the above condition is equivalent to
\begin{equation*}
    \big[ A(x) \Gamma + B(x) K \Gamma \big]^{\top} \Gamma^{-1} \big[ A(x) \Gamma + B(x) K \Gamma \big] - \Gamma \preceq -\epsilon_{\Gamma} I_n
\end{equation*}
for all $x \in \mathcal{B}_r$.
Pre- and post-multiplying the above inequality by $P = \Gamma^{-1}$ leads to
\begin{align*}
    \big[ A(x) + B(x)K \big]^{\top} P \big[ A(x) + B(x)K \big] - P & \preceq -\epsilon_{\Gamma} P^2 
    \\
    \preceq - \epsilon_{\Gamma} \lambda^{2}_{\min}(P) I_n 
    & \preceq - \frac{ \epsilon_{\Gamma} } { \lambda^{2}_{\max}(\Gamma) } I_n ,
\end{align*}
which implies
\begin{equation*}
    \big[ A(x) + B(x)K \big]^{\top} P \big[ A(x) + B(x)K \big] - P \preceq - \epsilon I_n 
    \ \ \forall x \in \mathcal{B}_r
\end{equation*}
with $\epsilon = { \epsilon_{\Gamma} } / { \lambda^{2}_{\max}(\Gamma) } > 0$, 
that is, \eqref{eq:LyaCondition} in Lemma~\ref{lema:modelBasedSolution} holds.
Moreover, the lower-right principal block of \eqref{eq:LMI_SufficientLyaCondition} implies $\Gamma-\epsilon_\Gamma I_n \succeq 0$, and hence $\epsilon_\Gamma\leq\lambda_{\min}(\Gamma)$. Since \eqref{eq:LMI_SufficientLyaCondition} remains valid for any smaller positive $\epsilon_\Gamma$, without loss of generality, we take $0<\epsilon_\Gamma<\lambda_{\min}(\Gamma)$, thus $\epsilon_\Gamma\lambda_{\min}(\Gamma) < \lambda_{\max}^2(\Gamma)$ and the denominator in \eqref{eq:LMI_r0} is positive.
As for the region of attraction $\mathcal{B}_{r_0}$, 
\eqref{eq:modelBased_r0} rewrites into \eqref{eq:LMI_r0} through $\lambda_{\max}(P) = \lambda^{-1}_{\min}(\Gamma)$ and $\lambda_{\min}(P) = \lambda^{-1}_{\max}(\Gamma)$.
This completes the proof.
\end{proof}

\subsection{Proof of Theorem~\ref{thm:Robust2Dist}}
\label{apx:thm:Robust2Dist}

\begin{proof}[Proof of Theorem~\ref{thm:Robust2Dist}]
In Theorem~\ref{thm:LMIsolution} we have proven that 
if there exist $\Gamma \succ 0$, $\epsilon_{\Gamma} > 0$ and $Y$ 
such that \eqref{eq:LMI_SufficientLyaCondition} holds,
then
\begin{equation}\label{eq:R2D_step0}
    \big[ \ast \big]^{\top} P \cdot \big[ A(x) + B(x)K \big] - P \preceq - \epsilon I_n 
    \quad \forall x \in \mathcal{B}_r
\end{equation}
with $K = Y \Gamma^{-1}$, $P = \Gamma^{-1}$ 
and $\epsilon = { \epsilon_{\Gamma} } / { \lambda^{2}_{\max}(\Gamma) } > 0$.
Let $V(x) := x^{\top} P x$ and $A_{\cl} := A(x) + B(x)K$.
Along with the trajectory of \eqref{eq:disturbed_ClosedLoop}, 
\begin{align}
    & V(x(k+1)) - V(x(k)) = V( A_{\cl} x(k) ) - V(x(k)) \notag \\
    & \hspace*{2cm} + V( A_{\cl} x(k) + w(k) ) - V( A_{\cl} x(k) ) \notag \\
    & = x(k)^{\top} \Big( A_{\cl}^{\top} P A_{\cl} - P \Big) x(k) \notag \\
    & \hspace*{2cm} + 2 x(k)^{\top} A_{\cl}^{\top} P w(k) + w(k)^{\top} P w(k) \notag \\
    & \overset{\eqref{eq:R2D_step0}}{\leq} - \epsilon \, |x(k)|^{2} 
    + 2 x(k)^{\top} A_{\cl}^{\top} P w(k) + w(k)^{\top} P w(k) \label{eq:R2D_step1}
\end{align}
for all $x(k) \in \mathcal{B}_r$.
Let $\gamma_{A_{\cl}} := \max_{x \in \mathcal{B}_r} \| A(x) + B(x)K \|$ and 
$| w |_{\infty} := \sup_{k \geq 0} |w(k)|$, then
\begin{align}
    \eqref{eq:R2D_step1} \leq & - \epsilon \, |x(k)|^{2} \notag \\
    & + 2 \gamma_{A_{\cl}} \lambda_{\max}(P) \, |x(k)| \, |w|_{\infty} + \lambda_{\max}(P) \, |w|_{\infty}^2 . \label{eq:R2D_step2}
\end{align}
For any $\alpha > 0$, the Young's inequality yields
\begin{align*}
    2 \, |x(k)| \, |w|_{\infty} \leq \frac{1}{\alpha} |x(k)|^2 + \alpha |w|_{\infty}^2 .
\end{align*}
By taking $\alpha := 2 \gamma_{A_{\cl}} \lambda_{\max}(P) / \epsilon$, we have
\begin{align}
    \eqref{eq:R2D_step2} \leq & - \frac{\epsilon}{2} \, |x(k)|^{2} \notag \\
    & + \underbrace{ \Big( 2 \gamma_{A_{\cl}}^2 \lambda^2_{\max}(P) / \epsilon + \lambda_{\max}(P) \Big) }_{ =: c_{w}} |w|_{\infty}^2 . \label{eq:R2D_step3}
\end{align}
For all $x(k) \in \mathcal{B}_r$, 
combining \eqref{eq:R2D_step1}, \eqref{eq:R2D_step2} and \eqref{eq:R2D_step3} gives
\begin{subequations}\label{eq:R2D_step4}
\begin{align}
    V(x(k+1)) - V(x(k)) & \leq - \frac{\epsilon}{2\lambda_{\max}(P)} \, V(x(k)) + c_{w} |w|_{\infty}^2 \notag \\
    \iff V(x(k+1)) & \leq \mu_{w} V(x(k)) + c_{w} |w|_{\infty}^2  \label{eq:R2D_step4-1}
\end{align} 
where
\begin{align}
    \mu_{w} := 1 - \frac{\epsilon}{2\lambda_{\max}(P)}.
\end{align}
\end{subequations}
Note that $0 < \epsilon \leq \lambda_{\max}(P)$, thus $0 < \mu_{w} < 1$.
Moreover, note that \eqref{eq:R2D_step4} only holds for $x(k) \in \mathcal{B}_r$.
Next we need to prove that, if \eqref{eq:R2D_condition} holds, 
then \eqref{eq:R2D_step4} holds for all $k \in \mathbb{N}$.
Bearing in mind $\lambda_{\max}(P) = \lambda^{-1}_{\min}(\Gamma)$, $\lambda_{\min}(P) = \lambda^{-1}_{\max}(\Gamma)$ and $\epsilon = { \epsilon_{\Gamma} } / { \lambda^{2}_{\max}(\Gamma) }$,
\eqref{eq:R2D_condition} equivalently rewrites as $|x(0)|^2 \leq r^2$ and
\begin{align}\label{eq:R2D_step5}
    \frac{\lambda_{\max}(P)}{\lambda_{\min}(P)} \mu_{w} \, |x(0)|^2 + \frac{1}{1-\mu_{w}} \frac{c_{w}}{\lambda_{\min}(P)} \, |w|^2_{\infty} \leq r^2 .
\end{align}
For $x(0) \in \mathcal{B}_r$, \eqref{eq:R2D_step4} holds at $k=0$, which implies
\begin{align*}
    |x(1)|^2 \leq & \frac{\lambda_{\max}(P)}{\lambda_{\min}(P)} \mu_{w} \, |x(0)|^2 
    + \frac{c_{w}}{\lambda_{\min}(P)} \, |w|^2_{\infty} 
    \overset{\eqref{eq:R2D_step5}}{\leq} r^2 ,
\end{align*}
and further implies $x(1) \in \mathcal{B}_r$.
Let $x(i) \in \mathcal{B}_{r}$ for $i = 1,\dots,k$; 
this condition together with \eqref{eq:R2D_step4} implies
\begin{align*}
    V(x(k+1)) \ \leq \ \ & \mu_{w}^{k+1} \, V(x(0)) \\
    & + \left(1+\mu_{w}+\dots+\mu_{w}^{k}\right) c_{w} \, |w|^2_{\infty} \\ 
    \implies |x(k+1)|^2 \ \leq \ \ & \frac{\lambda_{\max}(P)}{\lambda_{\min}(P)} \mu_{w}^{k+1} \, |x(0)|^2 \\
    & + \frac{1-\mu_{w}^{k+1}}{1-\mu_{w}} \frac{c_{w}}{\lambda_{\min}(P)} \, |w|^2_{\infty} \\
    \overset{\eqref{eq:R2D_step5}}{\implies} |x(k+1)|^2 \ \leq \ \ & r^2   \implies x(k+1) \in \mathcal{B}_r
\end{align*}
By induction, $x(k) \in \mathcal{B}_{r}$ for all $k \geq 1$.
So far we conclude that suppose \eqref{eq:LMI_SufficientLyaCondition} is feasible, 
if \eqref{eq:R2D_condition} holds, 
then \eqref{eq:R2D_step4} holds for all $k \in \mathbb{N}$, and
\begin{align*}
    |x(k)|^2  \leq  \frac{\lambda_{\max}(P)}{\lambda_{\min}(P)} \mu_{w}^{k} \, |x(0)|^2 
    + \frac{1}{1-\mu_{w}} \frac{c_{w}}{\lambda_{\min}(P)} \, |w|^2_{\infty}
\end{align*}
for all $k \in \mathbb{N}$, that is, \eqref{eq:R2D_result} holds. This completes the proof.
\end{proof}

\subsection{Proof of Lemma~\ref{lema:intermeLyaCondi}}
\label{apx:lema:intermeLyaCondi}

\begin{proof}[Proof of Lemma~\ref{lema:intermeLyaCondi}]
The condition \eqref{eq:LyaCondition_satur} implies that, 
for $x(k) \in \mathcal{B}_r \cap \mathcal{E}_{P, \tau}$,
\begin{align}\label{eq:LyaDecrease_satur}
    V(x(k+1)) \leq \mu V(x(k)) \ \text{ with } \ \mu := 1 - \frac{\epsilon}{\lambda_{\max}(P)},
\end{align} 
which corresponds to \eqref{eq:LyaDecrease}.
Without loss of generality, we take $0<\epsilon<\lambda_{\min}(P)\leq\lambda_{\max}(P)$, thus $\lambda_{\max}(P)-\epsilon>0$, the denominator in \eqref{eq:modelBased_r0_satur} is positive, and $0<\mu<1$.
Using analogous arguments as in the proof of Lemma~\ref{lema:modelBasedSolution},
we obtain that $x(k) \in \mathcal{B}_{r}$ for all $k \geq 1$ 
if $x(0) \in \mathcal{B}_{r_0}$ 
with $r_0$ specified in \eqref{eq:modelBased_r0_satur},
which is consistent with \eqref{eq:modelBased_r0}.
Moreover, for any $x(k) \in \mathcal{E}_{P, \tau}$,
\begin{align*}
    x(k) \in \mathcal{E}_{P, \tau} 
    & \overset{0 \leq \mu}{\iff} 
    \mu V(x(k)) \leq \mu \tau
    \overset{\eqref{eq:LyaDecrease_satur}}{\implies}
    V(x(k+1)) \leq \mu \tau
    \\
    & \overset{0 \leq \mu < 1}{\implies}
    V(x(k+1)) \leq \tau
    \iff x(k+1) \in \mathcal{E}_{P, \tau}.
\end{align*}
In other words,
$x(k) \in \mathcal{E}_{P, \tau}$ for all $k \geq 1$  if $x(0) \in \mathcal{E}_{P, \tau}$.
So far we conclude that, under condition \eqref{eq:LyaDecrease_satur}, 
if $x(0) \in \mathcal{B}_{r_0} \cap \mathcal{E}_{P, \tau}$, 
then $x(k) \in \mathcal{B}_{r} \cap \mathcal{E}_{P, \tau}$ for all $k \geq 1$.
This fact implies that 
\begin{equation*}
    V(x(k+1)) \leq \mu V(x(k)) \quad \forall k \geq 0,
\end{equation*}
with $0 \leq \mu < 1$,
which completes the proof.
\end{proof}

\subsection{Proof of Theorem~\ref{thm:LMIsolution_satur}}
\label{apx:thm:LMIsolution_satur}

\begin{proof}[Proof of Theorem~\ref{thm:LMIsolution_satur}]
Pre- and post-multiplying the inequalities in \eqref{eq:subsetCondi_satur} by $\blkdiag(\Gamma^{-1}, 1)$ yields, with $L = W \Gamma^{-1}$,
\[ 
    \bmat{\Gamma^{-1} & L_{(i)}^{\top} \\[3pt] L_{(i)} & \bar{u}_i^2} \succeq 0 
    \Leftrightarrow L_{(i)} \Gamma L_{(i)}^{\top} \leq \bar{u}_i^2
    \Leftrightarrow | L_{(i)} \Gamma^{1/2} | \leq \bar{u}_i .
\] 
For every $x \in \mathcal{E}_{P,1}$, we have $| P^{1/2} x | \leq 1$ 
and 
\[
    | L_{(i)} x | = | L_{(i)} \Gamma^{1/2} P^{1/2} x | 
    \leq | L_{(i)} \Gamma^{1/2} | \cdot | P^{1/2} x |
    = \bar{u}_i,
\]
which implies $x \in \mathcal{S}(L)$, namely, $\mathcal{E}_{P,1} \subseteq \mathcal{S}(L)$.
By Lemma~\ref{lema:sectorCondi}, we conclude that \eqref{eq:subsetCondi_satur} results in
\begin{equation}\label{eq:sectorCondi_thm}
    \phi(u)^{\top} N \big( \sat(u) + Lx \big) \leq 0 \quad \forall x \in \mathcal{E}_{P,1}
\end{equation}
with any diagonal positive definite matrix $N$. \\
Define the function
\[
    c(G) :=
    \bmat{
        -\Gamma \! + \epsilon_{\Gamma} I_n 
        & \ast 
        & \ast
        \\[3pt]
        -Y-W & -2 S & \ast 
        \\[3pt]
        G \bmat{\Gamma \\ Y} & G \bmat{0 \\ S} & -\Gamma 
    }.
\]
Condition \eqref{eq:stableCondi_satur} is equivalent to
\begin{align*}
    c(G) \preceq 0 \ \ \forall G \in \bar{\mathcal{G}} 
    \iff & c(G) \preceq 0 \ \ \forall G \in {\rm{conv}}(\bar{\mathcal{G}}) 
    \\
    \overset{\eqref{eq:G_subset_convG}}{\implies} & c(G) \preceq 0 \ \ \forall G \in \mathcal{G}
    \\
    \overset{\eqref{eq:GxOO}}{\iff} & c\big( [A(x) \ B(x)] \big) \preceq 0 \ \ \forall x \in \mathcal{B}_r .
\end{align*}
Taking $W = L \Gamma$, $N = S^{-1}$, $Y = K \Gamma$, 
pre- and post-multiplying $c\big( [A(x) \ B(x)] \big)$ by $\blkdiag(I_n, N, I_n)$ 
results in
\begin{equation*}
    \bmat{
        -\Gamma + \epsilon_{\Gamma} I_n & \ast & \ast
        \\
        -N (K \Gamma + L \Gamma) & -2N & \ast
        \\
        A(x) \Gamma + B(x) K \Gamma & B(x) & -\Gamma
    } \preceq 0 
    \quad \forall x \in \mathcal{B}_r .
\end{equation*}
Let $\epsilon = { \epsilon_{\Gamma} } / { \lambda^{2}_{\max}(\Gamma) }$,
the above inequality implies
\begin{equation*}
    \bmat{
        -\Gamma + \epsilon \Gamma^{2} & \ast & \ast
        \\
        -N (K \Gamma + L \Gamma) & -2N & \ast
        \\
        A(x) \Gamma + B(x) K \Gamma & B(x) & -\Gamma
    } \preceq 0 
    \quad \forall x \in \mathcal{B}_r .
\end{equation*}
Pre- and post-multiplying both sides of the above inequality by $\blkdiag(P, I_m, I_n)$ with $P = \Gamma^{-1}$,
and by Schur complement, the above inequality is equivalent to 
$H(x) \preceq 0$ for all $x \in \mathcal{B}_r$, where
\begin{align}
    & H(x) := \label{eq:H(x)}
    \\
    & \left[ \begin{matrix}
        \big( A(x)+B(x)K \big)^{\top} P \big( A(x)+B(x)K \big) - P + \epsilon I_n 
        \\
        B(x)^{\top} P \big( A(x)+B(x)K \big) - N (K+L)
    \end{matrix} \right. \notag
    \\
    & \hspace{20mm} 
    \left.\begin{matrix}
        \big( A(x)+B(x)K \big)^{\top} P B(x) - (K+L)^{\top} N
        \\
        B(x)^{\top} P B(x) - 2 N
    \end{matrix} \right]. \notag
\end{align}
This fact implies, for $V(x) := x^{\top} P x$,
\begin{align*}
    & \hspace{11mm} \bmat{x \\ \phi(Kx)}^{\top} H(x) \bmat{x \\ \phi(Kx)} \leq 0
    \quad \forall x \in \mathcal{B}_r
    \\
    & \overset{\eqref{eq:nSYS_satur_closed},\eqref{eq:sectorCondi}}{\iff} \
    V(x^+) - V(x) 
    \\
    & \hspace{12mm} - 2 \phi(Kx)^{\top} N \big( \sat(Kx) + Lx \big) \leq -\epsilon |x|^2
    \ \ \forall x \in \mathcal{B}_r .
\end{align*}
Because of \eqref{eq:subsetCondi_satur} and \eqref{eq:sectorCondi_thm}, we obtain
\begin{align*}
    V(x^+) - V(x) \leq - \epsilon |x|^2 \quad
    \forall x \in \mathcal{B}_r \cap \mathcal{E}_{P,1} ,
\end{align*}
which is consistent with \eqref{eq:LyaCondition_satur}. 
Without loss of generality, we take $0<\epsilon_\Gamma<\lambda_{\min}(\Gamma)$, thus $\epsilon_\Gamma\lambda_{\min}(\Gamma) < \lambda_{\max}^2(\Gamma)$, and the denominator in \eqref{eq:modelBased_r0_satur_thm} is positive.
Then by Lemma~\ref{lema:intermeLyaCondi}, 
and rewriting \eqref{eq:modelBased_r0_satur} to \eqref{eq:modelBased_r0_satur_thm}
via $\epsilon = { \epsilon_{\Gamma} } / { \lambda^{2}_{\max}(\Gamma) }$,
$\lambda_{\max}(P) = \lambda^{-1}_{\min}(\Gamma)$ and $\lambda_{\min}(P) = \lambda^{-1}_{\max}(\Gamma)$,
we complete the proof.
\end{proof}

\subsection{Proof of Lemma~\ref{lema:polytope_data}}
\label{apx:lema:polytope_data}

\begin{proof}[Proof of Lemma~\ref{lema:polytope_data}]
Assumption~\ref{ass:DisturbanceEnergy} ensures that the actual system is in the set $\mathcal{C}$, i.e.,
$[ E_A \ E_B ] \in \mathcal{C}$, 
thus \eqref{eq:LMI_SufficientLyaCondition_data} implies 
\begin{align}
    & \bmat{ 
        -\Gamma & [E_A \ E_B] Q \bmat{\Gamma \\ Y}
        \\
        \ast & - \Gamma + \epsilon_{\Gamma} I_n
    } \preceq 0 
    \ \ \forall Q \in \bar{\mathcal{Q}} \notag
    \\ 
    & \Updownarrow \notag
    \\
    & \bmat{ 
        -\Gamma & \! \smat{E_A & E_B} \smat{Q_A \Gamma \\ Q_B Y}
        \\
        \ast & - \Gamma + \epsilon_{\Gamma} I_n
    } \preceq 0 
    \ \forall \blkdiag(Q_A, Q_B) \! \in \! {\rm{conv}}(\bar{\mathcal{Q}}) \notag
    \\
    & \Downarrow \text{\scriptsize \eqref{eq:Q_subset_convQ}} \notag
    \\
    & \bmat{ 
        -\Gamma & E_A \Xi_A(x) \Gamma + E_B \Xi_B(x) Y 
        \\
        \ast & - \Gamma + \epsilon_{\Gamma} I_n
    } \preceq 0 
    \ \ \forall x \in \mathcal{B}_r \notag
    \\
    & \Updownarrow \text{\scriptsize \eqref{eq:nSYSrewrite}} \notag
    \\
    & \bmat{ 
        -\Gamma & A(x) \Gamma + B(x) Y 
        \\
        \ast & - \Gamma + \epsilon_{\Gamma} I_n
    } \preceq 0 
    \ \ \forall x \in \mathcal{B}_r , 
    \label{eq:Lema_data_step1}
\end{align}
which corresponds to \eqref{eq:thm_model_step1} in the proof of Theorem~\ref{thm:LMIsolution}.
Without loss of generality, we take $0<\epsilon_\Gamma<\lambda_{\min}(\Gamma)$, thus the denominator in \eqref{eq:LMI_r0_data} is positive.
Then following analogous arguments as in the proof of Theorem~\ref{thm:LMIsolution}, we can complete the proof of this lemma.
\end{proof}

\subsection{Proof of Lemma~\ref{lema:petersen_data}}
\label{apx:lema:petersen_data}

\begin{proof}[Proof of Lemma~\ref{lema:petersen_data}]
By Assumption~\ref{ass:DataFullRow} and Schur complement, 
\eqref{eq:petersenLMIGxu} is equivalent to
\begin{align*}
    0 \succeq 
    \bmat{ 
        -\bar{\Gamma} - \mathbf{C} & \! 0 \!
        \\ 
        \ast \! & \! - \bar{\Gamma} + \epsilon_{\bar{\Gamma}} I_n 
    }  
    + \bmat{
        \mathbf{B} 
        \\ 
        -\smat{\bar{\Gamma} \\ \bar{Y}}^{\top} \!\! Q^{\top}
    } \cdot \mathbf{A}^{-1} \Big[\ast\Big]^{\top}
\end{align*}
for all $Q \in \bar{\mathcal{Q}}$.
Bearing in mind the definitions of $\mathbf{Z}_c$ and $\mathbf{Q}$ as in \eqref{eq:setC:ZQ_Gxu}, 
the above inequality equivalently rewrites as
\begin{align*}
    0 \! \succeq \! \!
    \bmat{ 
        -\bar{\Gamma} 
        & 
        \hspace*{-1mm} \mathbf{Z}_c
        Q \smat{\bar{\Gamma} \\ \bar{Y}} 
        \\[10pt] 
        \ast 
        & 
        \hspace*{-2mm} - \bar{\Gamma} + \epsilon_{\bar{\Gamma}} I_n
    }  
    \! + \! \bmat{
        \mathbf{Q} & \hspace*{-1mm} 0 
        \\ 
        0 & \hspace*{-1mm} 0
    }   
    \! + \! \bmat{
        0 & \hspace*{-1mm} 0 
        \\ 
        0 & \hspace*{-1mm} 
        \smat{\bar{\Gamma} \\ \bar{Y}}^{\top} \!\! Q^{\top} 
        \hspace*{-1.0mm}  \mathbf{A}^{\!-1} \hspace*{-0.2mm} 
        Q \smat{\bar{\Gamma} \\ \bar{Y}}
    }. 
\end{align*}
Multiplying both sides by $\tau^{-1} > 0$ and letting 
\begin{equation}\label{eq:petersen_step_tau}
    \Gamma = \tau^{-1} \bar{\Gamma}, \quad 
    \epsilon_{\Gamma} = \tau^{-1} \epsilon_{\bar{\Gamma}}, \quad
    Y = \tau^{-1} \bar{Y},  
\end{equation} 
the above inequality is equivalent to
\begin{align}
    0 \succeq 
    \bmat{ 
        -\Gamma & \mathbf{Z}_c Q \smat{\Gamma \\ Y} 
        \\[10pt] 
        \ast & - \Gamma + \epsilon_{\Gamma} I_n
    }  
    & + \frac{1}{\tau} \bmat{\mathbf{Q} & 0 \\ 0 & 0} \label{eq:petersen_step_2}  
    \\
    & + \tau \bmat{
        0 & 0 
        \\ 
        0 & 
        \smat{\Gamma \\ Y}^{\top} \!\! Q^{\top} 
        \hspace*{-1.0mm}  \mathbf{A}^{\!-1} \hspace*{-0.2mm} 
        Q \smat{\Gamma \\ Y}
    }. \notag
\end{align}
By Petersen's Lemma reported in~\cite[Fact~2]{AndreaPetersen2022}, 
the previous condition implies
\begin{align}
    0 & \succeq \! 
    \bmat{ 
        -\Gamma & \hspace*{-1mm} \mathbf{Z}_c Q \smat{\Gamma \\ Y}
        \\[10pt] 
        \ast & \hspace*{-1mm} - \Gamma + \epsilon_{\Gamma} I_n
    }  
    \! + \! \bmat{\mathbf{Q}^{1/2} \\ 0} 
    \! \Upsilon \! 
    \bmat{0 & \hspace*{-1mm} \mathbf{A}^{-1/2} Q \smat{\Gamma \\ Y}} 
    \! \label{eq:petersen_step_1} 
    \\
    & + \bmat{0 \\ \smat{\Gamma \\ Y}^{\!\top} \! Q^{\top} \mathbf{A}^{-1/2}} 
    \Upsilon^{\top} 
    \bmat{\mathbf{Q}^{1/2} & \hspace*{-1mm} 0} 
    \ \, \forall \Upsilon \text{ with } \Upsilon \Upsilon^\top \preceq  I_n . \notag
\end{align}
According to \eqref{eq:DataConsistentSetC_re}, \eqref{eq:petersen_step_1} is equivalent to
\begin{equation*}
    \bmat{ 
        -\Gamma & Z Q \bmat{\Gamma \\ Y} 
        \\[10pt] 
        \ast & - \Gamma + \epsilon_{\Gamma} I_n
    } \preceq 0 
    \quad \forall Z \in \mathcal{C},
\end{equation*}\
which corresponds to \eqref{eq:LMI_SufficientLyaCondition_data}.
This completes the proof.
\end{proof}

\subsection{Proof of Theorem~\ref{thm:LMIsolution_satur_data}}
\label{apx:thm:LMIsolution_satur_data}

\begin{proof}[Proof of Theorem~\ref{thm:LMIsolution_satur_data}]
Following analogous arguments as in the proof of Theorem~\ref{thm:LMIsolution_satur},
inequalities \eqref{eq:subsetCondi_satur_data} imply 
$\mathcal{E}_{\bar{P},1} \subseteq \mathcal{S}(L)$ with 
$L = \bar{W} \bar{\Gamma}^{-1}$.
By Lemma~\ref{lema:sectorCondi}, we conclude that \eqref{eq:subsetCondi_satur_data} results in
\begin{equation}\label{eq:sectorCondi_thm_data}
    \phi(u)^{\top} N \big( \sat(u) + Lx \big) \leq 0 \quad \forall x \in \mathcal{E}_{\bar{P},1}
\end{equation}
for any diagonal positive definite matrix $N$. \\
By Schur complement and the definitions of $\mathbf{Z}_c$ and $\mathbf{Q}$ in \eqref{eq:setC:ZQ_Gxu}, 
the linear matrix inequality in \eqref{eq:stableCondi_satur_data} is equivalent to
\begin{align*}
    0 \succeq 
    & \bmat{
        -\bar{\Gamma} + \epsilon_{\bar{\Gamma}} I_n
        & \ast
        & \ast
        \\[0pt]
        - \bar{Y} - \bar{W} 
        & -2 \bar{S}
        & \ast
        \\[0pt]
        \mathbf{Z}_c Q \smat{\bar{\Gamma} \\ \bar{Y}}
        & \mathbf{Z}_c Q \smat{ 0 \\ \bar{S}}
        & -\bar{\Gamma}  
    }
    + 
    \bmat{
        0 & \ast & \ast \\
        0 & 0 & \ast \\
        0 & 0 & \mathbf{Q}
    }
    \\
    & +
    \bmat{
        \smat{ \bar{\Gamma} \\ \bar{Y}}^{\! \top} \! Q^{\top}
        \! \mathbf{A}^{-1}
        Q \smat{\bar{\Gamma} \\ \bar{Y}}
        & \ast
        & \ast
        \\[5pt]
        \smat{ 0 \\ \bar{S}}^{\! \top} \! Q^{\top}
        \! \mathbf{A}^{-1} 
        Q \smat{\bar{\Gamma} \\ \bar{Y}}
        & \smat{ 0 \\ \bar{S}}^{\! \top} \! Q^{\top}
        \! \mathbf{A}^{-1}
        Q \smat{ 0 \\ \bar{S}}
        & \ast
        \\[5pt]
        0 & 0 & 0
    }.
\end{align*}
Multiplying both sides by $\tau^{-1} > 0$ and letting 
\begin{equation}\label{eq:petersen_step_tau_data}
\begin{aligned}
    & 
    \Gamma = \tau^{-1} \bar{\Gamma}, \quad
    \epsilon_{\Gamma} = \tau^{-1} \epsilon_{\bar{\Gamma}}, \quad
    Y = \tau^{-1} \bar{Y}, 
    \\  
    & \hspace{10mm}
    W = \tau^{-1} \bar{W}, \quad
    S = \tau^{-1} \bar{S}, 
\end{aligned}
\end{equation} 
the above matrix inequality is equivalent to
\begin{align*}
    0 \succeq 
    & \bmat{
        -\Gamma + \epsilon_{\Gamma} I_n
        & \ast
        & \ast
        \\[0pt]
        - Y - W 
        & -2 S
        & \ast
        \\[0pt]
        \mathbf{Z}_c Q \smat{\Gamma \\ Y}
        & \mathbf{Z}_c Q \smat{0 \\ S}
        & - \Gamma
    }
    +
    \frac{1}{\tau}
    \bmat{
        0 & \ast & \ast \\
        0 & 0 & \ast \\
        0 & 0 & \mathbf{Q}
    }
    \\
    & +
    \tau
    \bmat{
        \smat{\Gamma \\ Y}^{\! \top} \! Q^{\top} 
        \! \mathbf{A}^{-1}
        Q \smat{\Gamma \\ Y}
        & \ast
        & \ast
        \\[5pt]
        \smat{ 0 \\ S}^{\! \top} \! Q^{\top} 
        \! \mathbf{A}^{-1}
        Q \smat{\Gamma \\ Y}
        & \smat{ 0 \\ S}^{\! \top} \! Q^{\top} 
        \! \mathbf{A}^{-1}
        Q \smat{ 0 \\ S}
        & \ast
        \\[5pt]
        0 & 0 & 0
    }.
\end{align*}
By Petersen's Lemma reported in~\cite[Fact~2]{AndreaPetersen2022}, 
the previous inequality implies
\begin{align*}
    0 \succeq 
    & \bmat{
        -\Gamma + \epsilon_{\Gamma} I_n
        & \ast
        & \ast
        \\[0pt]
        - Y - W 
        & -2 S
        & \ast
        \\[0pt]
        \mathbf{Z}_c Q \smat{\Gamma \\ Y}
        & \mathbf{Z}_c Q \smat{0 \\ S}
        & - \Gamma
    }
    \\
    & + 
    \bmat{
        0 & \ast & \ast
        \\
        0 & 0 & \ast 
        \\
        \mathbf{Q}^{1/2} \Upsilon \mathbf{A}^{-1/2} Q \smat{\Gamma \\ Y}
        & \mathbf{Q}^{1/2} \Upsilon \mathbf{A}^{-1/2} Q \smat{ 0 \\ S}
        & 0
    }
\end{align*}
for all $\Upsilon$ with $\Upsilon \Upsilon^\top \preceq  I_n$.
Under Assumption~\ref{ass:DataFullRow} and \eqref{eq:DataConsistentSetC_re}, 
the above condition is equivalent to
\begin{align*}
    \bmat{
        -\Gamma + \epsilon_{\Gamma} I_n
        & \ast
        & \ast
        \\[0pt]
        - Y - W 
        & -2 S
        & \ast
        \\[0pt]
        Z Q \smat{\Gamma \\ Y}
        & Z Q \smat{ 0 \\ S}
        & - \Gamma
    } \preceq 0
    \quad \forall Z \in \mathcal{C}.
\end{align*}
Note that Assumption~\ref{ass:DisturbanceEnergy} ensures that $[ E_A \ E_B ] \in \mathcal{C}$. 
So far, by the above conditions, we have established that \eqref{eq:stableCondi_satur_data} implies that
\begin{align*}
    \bmat{
        -\Gamma + \epsilon_{\Gamma} I_n
        & \ast
        & \ast
        \\[0pt]
        - Y - W 
        & -2 S
        & \ast
        \\[0pt]
        E_A Q_A \Gamma + E_B Q_B Y  
        & E_B Q_B S
        & - \Gamma
    } \preceq 0
\end{align*}
for all $\blkdiag(Q_A, Q_B) \in \bar{\mathcal{Q}}$,
which further implies, by \eqref{eq:Q_subset_convQ} and \eqref{eq:nSYSrewrite}, that
\begin{align*}
    \bmat{
        -\Gamma + \epsilon_{\Gamma} I_n
        & \ast
        & \ast
        \\[0pt]
        - Y - W 
        & -2 S
        & \ast
        \\[0pt]
        A(x) \Gamma + B(x) Y  
        & B(x) S
        & - \Gamma
    } \preceq 0
     \quad \forall x \in \mathcal{B}_r.
\end{align*}
Let
\begin{equation}
\begin{aligned}
    & N = S^{-1}, \quad 
    Y = K \Gamma, \quad
    W = L \Gamma, \quad
    \\
    & \hspace{4mm} P = \Gamma^{-1}, \quad
    \epsilon = { \epsilon_{\Gamma} } / { \lambda^{2}_{\max}(\Gamma) }.
\end{aligned}
\end{equation}
Pre- and post-multiplying both sides of the above inequality by $\blkdiag(P, N, I_n)$ results in
\begin{align*}
    \bmat{
        -P + \epsilon I_n
        & \ast
        & \ast
        \\[0pt]
        - N( K + L ) 
        & -2 N
        & \ast
        \\[0pt]
        A(x)  + B(x) K  
        & B(x)
        & - P^{-1}
    } \preceq 0
     \quad \forall x \in \mathcal{B}_r.
\end{align*}
By Schur complement, the above inequality is equivalent to 
$H(x) \preceq 0$ for all $x \in \mathcal{B}_r$, 
where $H(x)$ has the same expression as in \eqref{eq:H(x)}.
This fact implies, for $V(x) := x^{\top} P x$, that
\begin{align*}
    & \hspace{11mm} \bmat{x \\ \phi(Kx)}^{\top} H(x) \bmat{x \\ \phi(Kx)} \leq 0
    \quad \forall x \in \mathcal{B}_r
    \\
    & \overset{\eqref{eq:nSYS_satur_closed},\eqref{eq:sectorCondi}}{\iff} \
    V(x^+) - V(x) 
    \\
    & \hspace{12mm} - 2 \phi(Kx)^{\top} N \big( \sat(Kx) + Lx \big) \leq -\epsilon |x|^2
    \ \ \forall x \in \mathcal{B}_r .
\end{align*}
Because of \eqref{eq:subsetCondi_satur_data} and \eqref{eq:sectorCondi_thm_data}, we obtain
\begin{align*}
    & V(x^+) - V(x) \leq - \epsilon |x|^2 \quad
    \forall x \in \mathcal{B}_r \cap \mathcal{E}_{\bar{P},1}
    \\
    & \Updownarrow  P = \tau \bar{P} \text{ and } \mathcal{E}_{\bar{P},1} = \mathcal{E}_{P,\tau}
    \\
    & V(x^+) - V(x) \leq - \epsilon |x|^2 \quad
    \forall x \in \mathcal{B}_r \cap \mathcal{E}_{P,\tau} ,
\end{align*}
which is consistent with \eqref{eq:LyaCondition_satur}.
Without loss of generality, we take $0<\epsilon_{\bar\Gamma}<\lambda_{\min}(\bar\Gamma)$, thus the denominator in \eqref{eq:modelBased_r0_satur_thm_data} is positive.
Then by Lemma~\ref{lema:intermeLyaCondi}, 
and rewriting \eqref{eq:modelBased_r0_satur} to \eqref{eq:modelBased_r0_satur_thm_data}
via $\epsilon = { \epsilon_{\Gamma} } / { \lambda^{2}_{\max}(\Gamma) }$,
$\lambda_{\max}(P) = \lambda^{-1}_{\min}(\Gamma)$, 
$\lambda_{\min}(P) = \lambda^{-1}_{\max}(\Gamma)$,
and \eqref{eq:petersen_step_tau_data},
we complete the proof.
\end{proof}

\section{Additional Simulation Experiments}
\label{apx:addiExp}

\subsection{Model-Based Simulations}

We first consider the model-based case.
The nonlinear system is \eqref{eq:example_sys}
with the corresponding state-dependent representation
\begin{equation*}
    A(x) \! = \!\! \bmat{ 
        1 + 0.1 \,  \sinc x_1  \!\! & \!\! 0.2 
        \\
        0.2 \!\! & \!\! 0.9 + 0.1 x_1^2
    } \! ,
    B(x) \! = \!\! \bmat{
        0.1 + 0.1 |x_2|
        \\
        0.1 e^{x_1}
    }
\end{equation*}
Take $r = 1.1$ for $\mathcal{B}_r$.
In \eqref{eq:min_max_gji}, 
$\bar{a}_{11} := \max_{x \in \mathcal{B}_r} a_{11}(x) = 1.1$, 
$\underline{a}_{11} := \min_{x \in \mathcal{B}_r} a_{11}(x) = 1.081$, 
$\bar{a}_{21} = \underline{a}_{21} = 0.2$, 
$\bar{a}_{12} = \underline{a}_{12} = 0.2$,
$\bar{a}_{22} = 1.021$, 
$\underline{a}_{22} = 0.9$, 
$\bar{b}_{11} = 0.21$, 
$\underline{b}_{11} = 0.1$, 
$\bar{b}_{21} = 0.3004$, 
$\underline{b}_{21} = 0.0332$, 
which lead to the set $\bar{\mathcal{G}}$ as in \eqref{eq:bar_set_G}.
By Theorem~\ref{thm:LMIsolution}, 
solving \eqref{eq:LMI_SufficientLyaCondition} via YALMIP \cite{YALMIP} yields 
$\epsilon_{\Gamma} = 0.1402$, $Y = [-65.0340 \  -78.9831]$, 
$\Gamma = \smat{34.3888 & 0.0000 \\ 0.0000 & 34.3888}$, and 
$K =  [-1.8911 \ -2.2968]$. 
According to \eqref{eq:LMI_r0}, 
the estimate of ROA, $\mathcal{B}_{r_0}$, has a radius of $r_0 = 1.1$.
Left subfigure of
Figure~\ref{fig:model_phase_portrait} is the phase portrait of the closed-loop system $x^{+} = A(x)x + B(x)Kx$.
We further investigate the relationship between $r$ and $r_0$. Keeping all other settings unchanged, we vary $r$ over $\{0.5,0.55,\ldots,1.5\}$. For each value of $r$, the controller is solved anew, and the corresponding ROA radius $r_0$ is computed. The results are shown in the right subfigure of Figure~\ref{fig:model_phase_portrait} and are consistent with the discussion in Section~\ref{sec:roa}.

\begin{figure}[t]
    \centering
    \subfigure{\includegraphics[width=0.49\columnwidth]{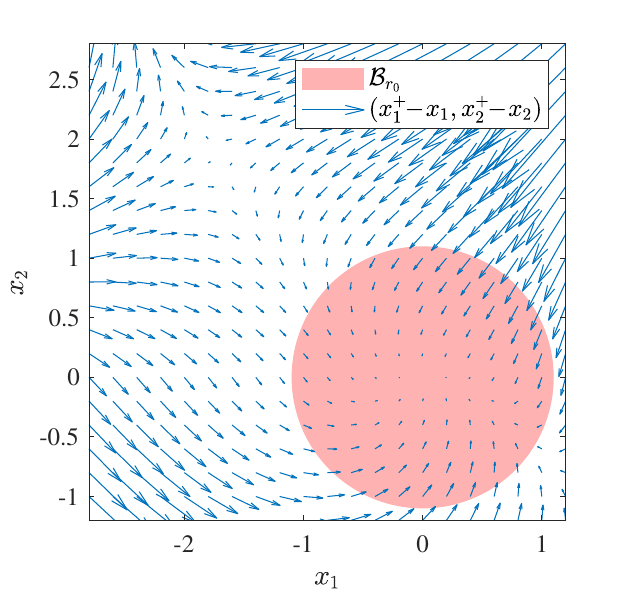}} 
    \subfigure{\includegraphics[width=0.49\columnwidth]{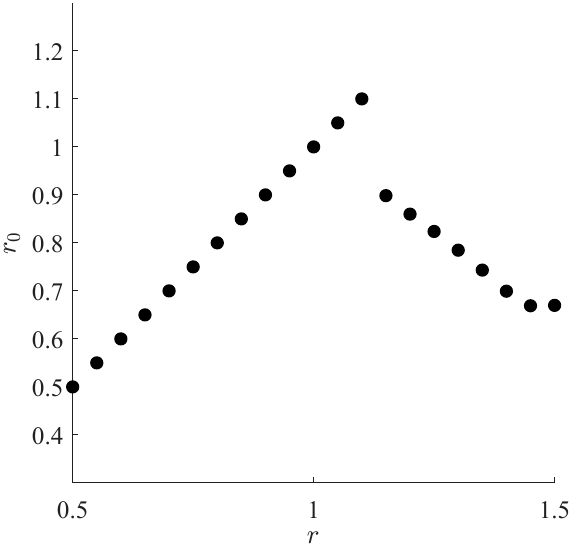}} 
    \vspace{-5mm}
    \caption{Left: phase portrait of the closed-loop system with model-based controller.
    Right: model-based region-of-attraction radius $r_0$ versus the prescribed state-region radius $r$.}
    \label{fig:model_phase_portrait}
    \vspace{-2mm}
\end{figure}

\begin{figure}[t]
    \centering
    \includegraphics[width=0.85\linewidth]{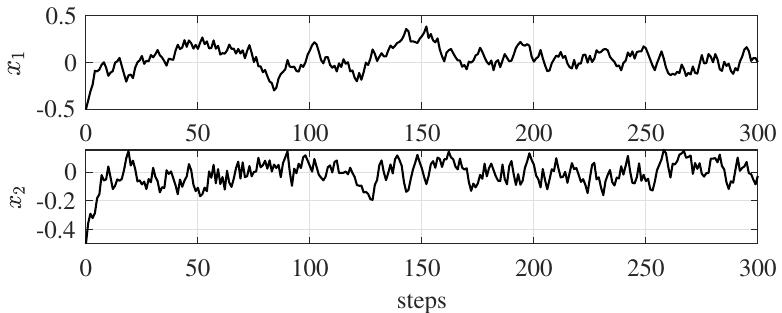}  
    \vspace{-2mm}
    \caption{State response of the closed-loop system with model-based controller under disturbance.}
    \label{fig:model_state_response}
    \vspace{-4mm}
\end{figure}

The controller $K$ remains unchanged.
For the closed-loop system $x^{+} = \big( A(x) + B(x)K \big) x + w$,
we set $x(0) = (-0.5, -0.5)$ and the disturbance $(w_1, w_2)$ as random variables uniformly distributed in the interval $[-0.1,0.1]$.
The simulation result is depicted in Figure~\ref{fig:model_state_response}, 
which is consistent with \eqref{eq:R2D_result} in Theorem~\ref{thm:Robust2Dist}.

Take $r = 1.1$ for $\mathcal{B}_r$. 
By \eqref{eq:min_max_gji}, we have the set $\bar{\mathcal{G}}$ as in \eqref{eq:bar_set_G}.
For different level of the saturation $\bar{u} \in \{4.0, \, 2.0, \, 1.0, \, 0.5\}$,
solving \eqref{eq:LMI_SufficientLyaCondition_satur} in Theorem~\ref{thm:LMIsolution_satur} via YALMIP yields 
$K = [-1.8155 \ -2.3147]$, $[-1.8155 \ -2.3146]$, $[-1.8155 \ -2.3147]$ and $[-1.9241 \ -2.3214]$, respectively.
The program \eqref{eq:LMI_SufficientLyaCondition_satur} is only a feasibility problem, 
and we set minimizing 
$[\lambda_{\max}(\Gamma) - \lambda_{\min}(\Gamma) - \epsilon_{\Gamma} - 0.1\text{trace}(\Gamma)]$
as the objective when solving the LMIs \eqref{eq:LMI_SufficientLyaCondition_satur}.
The estimate of ROA, $\mathcal{B}_{r_0} \cap \mathcal{E}_{P,1}$, 
for different values of $\bar u$ is shown in Figure~\ref{fig:model_satur_u}.

\begin{figure}[t]
    \centering
    \subfigure{\includegraphics[width=0.42\columnwidth]{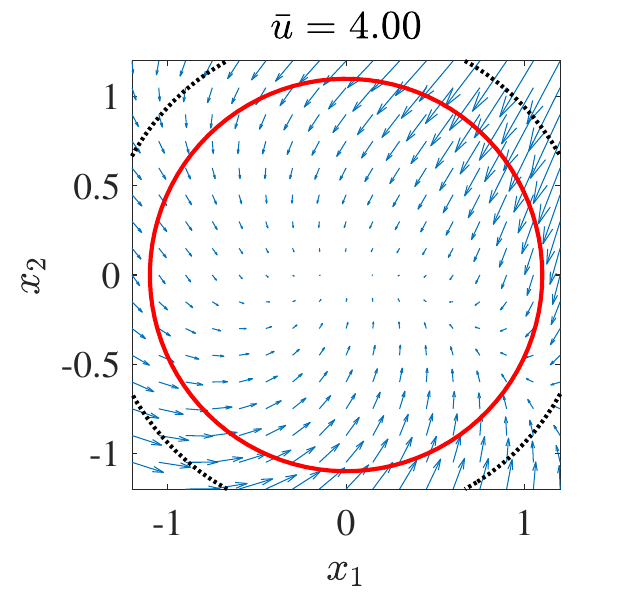}} 
    \subfigure{\includegraphics[width=0.42\columnwidth]{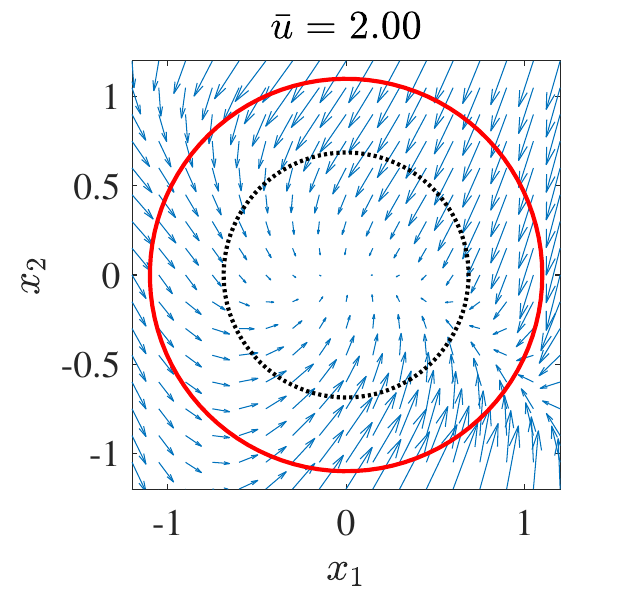}} \\[-6pt]
    \subfigure{\includegraphics[width=0.42\columnwidth]{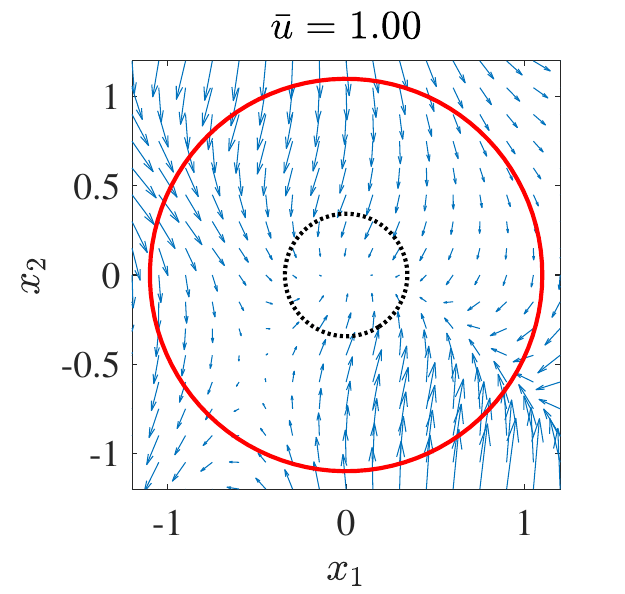}} 
    \subfigure{\includegraphics[width=0.42\columnwidth]{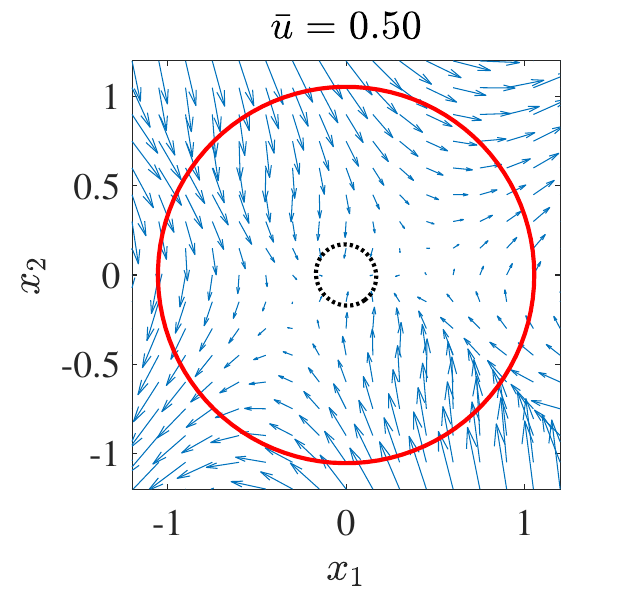}} 
    \vspace{-3mm}
    \caption{Estimate of ROA, $\mathcal{B}_{r_0} \cap \mathcal{E}_{P,1}$, for different saturation levels.
    The red solid ball is $\mathcal{B}_{r_0}$; 
    the black dashed ellipsoid is $\mathcal{E}_{P,1}$;
    the background is the phase portrait regarding the closed-loop system $x^{+} = A(x)x + B(x) \sat(Kx)$.}
    \label{fig:model_satur_u}
    \vspace{-4mm}
\end{figure}

\subsection{Data-Driven Simulations}
\label{apx:addiExpDDC}

Next we consider the data-driven case. The known vector-valued continuous functions are
\begin{align*}
    & \xi_{A_1}(x) = \left( 1, \  \sinc x_1  \right), \ \
    \xi_{A_2}(x) = \left( 1, \ x_1, \ x_1^2 \right),
    \\
    & \hspace{8mm} \xi_{B_1}(x) = \left( 1,\ |x_1|,\ |x_2|,\ e^{x_1},\ e^{x_2} \right).
\end{align*}
Assumption~\ref{ass:function_qi} allows us to write the state-dependent representation as 
$A(x) = E_A \blkdiag( \xi_{A_1}(x), \xi_{A_2}(x) )$ and 
$B(x) = E_B \xi_{B_1}(x)$ with
\begin{align*}
    E_A = \bmat{
        1 \! & \! 0.1 \! & \! 0.2 \! & \! 0 \! & \! 0 \\ 
        0.2 \! & \! 0 \! & \! 0.9 \! & \! 0 \! & \! 0.1
    },
    E_B = \bmat{
        0.1 \! & \! 0 \! & \! 0.1 \! & \! 0 \! & \! 0 \\
        0 \! & \! 0 \! & \! 0 \! & \! 0.1 \! & \! 0
    }.
\end{align*}
Note that the system knowledge $E_A$ and $E_B$ are unknown in the controller design.
We use the same dataset and noise as in Section~\ref{sec:comparative}.
Take $r = 0.92$ for $\mathcal{B}_r$. 
By \eqref{eq:min_max_qji}, we have the set $\bar{\mathcal{Q}}$ as in \eqref{eq:bar_set_Q}.
Solving \eqref{eq:petersenLMIGxu} in Theorem~\ref{thm:dataDrivenSolution} via YALMIP yields  
$\epsilon_{\bar \Gamma} = 0.000045$, $\bar Y = [-0.0178 \ -0.0186]$, 
$\bar \Gamma = \smat{0.0082 & 0.0000 \\ 0.0000 & 0.0082}$, and 
$K =  [-2.1797 \ -2.2750]$.
According to \eqref{eq:LMI_r0_data_bar}, 
the estimate of ROA, $\mathcal{B}_{\bar{r}_0}$, has a radius of $\bar{r}_0 = 0.92$.
Left subfigure of Figure~\ref{fig:data_phase_portrait} is the phase portrait of the closed-loop system $x^{+} = A(x)x + B(x)Kx$.
We further investigate the relationship between $r$ and $\bar r_0$. Keeping the experimental data set and all other settings unchanged, we vary $r$ over $\{0.5,0.54,\ldots,1.3\}$. For each value of $r$, the data-driven controller is solved anew using the same data set, and the corresponding ROA radius $\bar r_0$ is computed. The results are shown in the right subfigure of Figure~\ref{fig:data_phase_portrait} and are consistent with the discussion in Section~\ref{sec:roa}.

The data-driven controller $K$ remains unchanged.
For the closed-loop system $x^{+} = A(x)x + B(x)Kx + w$,
we set $x(0) = (-0.4, -0.4)$ and the disturbance $(w_1, w_2)$ as random variables uniformly distributed in the interval $[-0.1,0.1]$.
In Remark~\ref{rmk:gammaAcl}, the bound $\bar{\gamma}_{A_{\cl}}$ involves a maximization over $x\in\mathcal{B}_r$, which is approximated here by sampling, yielding $\bar{\gamma}_{A_{\cl}}=11.9367$.
For comparison, using the true model for a posteriori evaluation gives $\gamma_{A_{\cl}}=0.95527$.
Also, $\bar{\delta}_{x_0}=0.99725$; the disturbance coefficient is $\bar{\delta}_w=121308$ using $\gamma_{A_{\cl}}$ and increases to $\bar{\delta}_w=18884821$ using the data-driven bound $\bar{\gamma}_{A_{\cl}}$.
The conservatism is expected since \eqref{eq:R2D_condition_data} is only a sufficient condition that guarantees robustness for all systems consistent with the data.
As a result, the certified disturbance level is smaller than what the closed loop actually tolerates.
Indeed, Figure~\ref{fig:data_state_response} shows that the system tolerates disturbances larger than those certified by \eqref{eq:R2D_condition_data}.
A tighter data-driven bound would require a more refined robustness analysis and is left for future work.

The experimental data $X_0$, $X_1$, $U_0$ and the energy bound $\Theta$ remain unchanged.
Take $r = 0.92$ for $\mathcal{B}_r$. 
By \eqref{eq:min_max_gji}, we have the set $\bar{\mathcal{Q}}$ as in \eqref{eq:bar_set_Q}.
For different level of the saturation $\bar{u} \in \{2.0, \, 1.0, \, 0.5, \, 0.25\}$,
solving \eqref{eq:LMI_SufficientLyaCondition_satur_data} in Theorem~\ref{thm:LMIsolution_satur_data} via YALMIP yields 
$K = [-2.5369 \ -2.9559]$, $[-2.5369 \ -2.9559]$, $[-2.5369 \ -2.9559]$ and $[-2.6885 \ -2.2310]$, respectively.
The program \eqref{eq:LMI_SufficientLyaCondition_satur_data} is a feasibility problem, 
and we set minimizing 
$[\lambda_{\max}(\bar{\Gamma}) - \lambda_{\min}(\bar{\Gamma}) - 2\text{trace}(\bar{\Gamma})]$
as the objective when solving the LMIs \eqref{eq:LMI_SufficientLyaCondition_satur_data}.
The estimate of ROA, $\mathcal{B}_{\bar{r}_0} \cap \mathcal{E}_{\bar{P},1}$, 
for different values of $\bar u$ is shown in Figure~\ref{fig:data_satur_u}.

\begin{figure}[t]
    \centering
    \subfigure{\includegraphics[width=0.49\columnwidth]{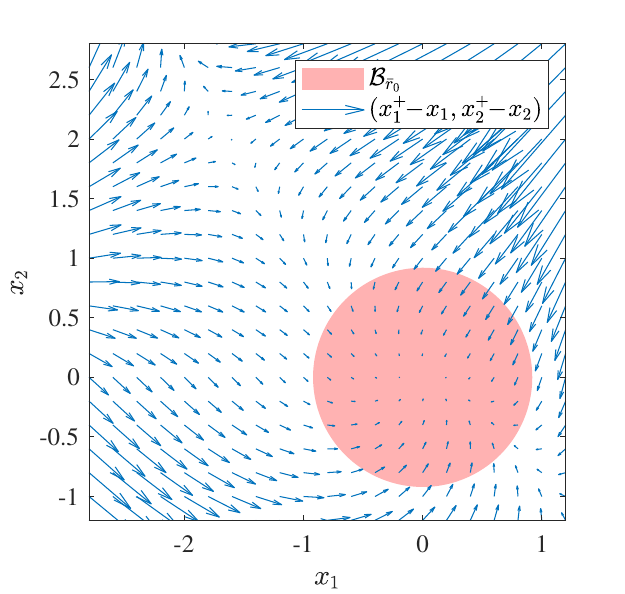}} 
    \subfigure{\includegraphics[width=0.49\columnwidth]{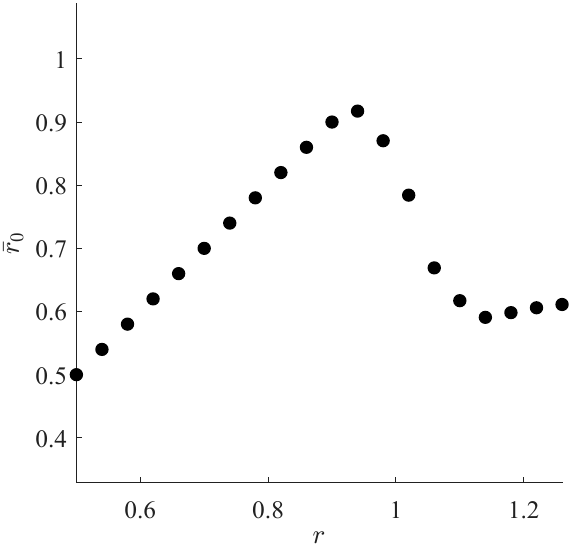}} 
    \vspace{-5mm}
    \caption{Left: phase portrait of the closed-loop system with model-based controller. 
    Right: data-driven region-of-attraction radius $\bar r_0$ versus the prescribed state-region radius $r$.}
    \label{fig:data_phase_portrait}
    \vspace{-2mm}
\end{figure}

\begin{figure}[t]
    \centering
    \includegraphics[width=0.85\linewidth]{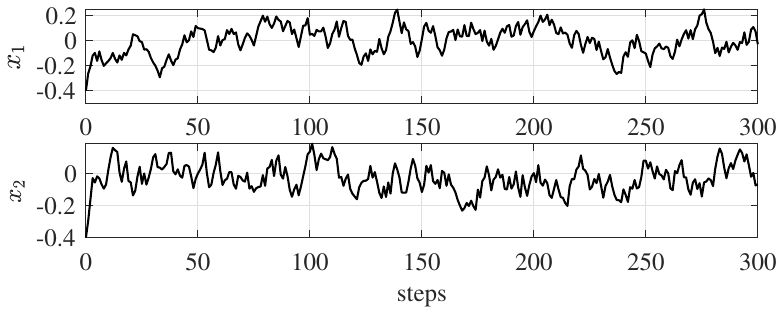}  
    \vspace{-2mm}
    \caption{State response of the closed-loop system with data-driven controller under disturbance.}
    \label{fig:data_state_response}
    \vspace{-4mm}
\end{figure}

\begin{figure}[t]
    \centering
    \subfigure{\includegraphics[width=0.42\columnwidth]{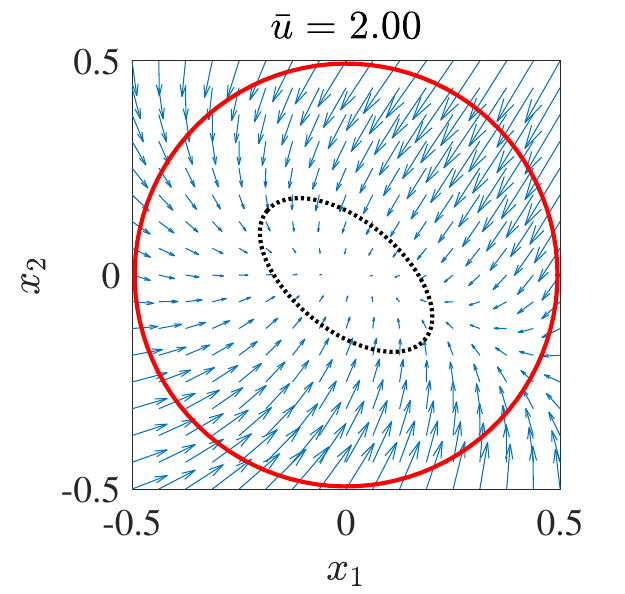}} 
    \subfigure{\includegraphics[width=0.42\columnwidth]{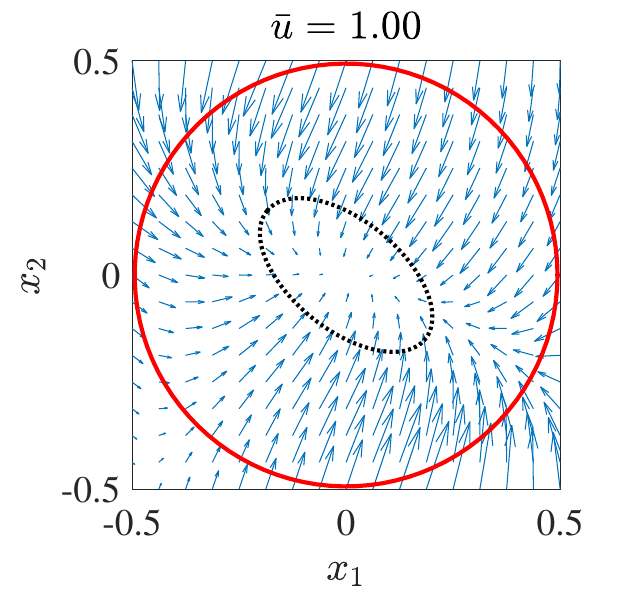}} \\[-6pt]
    \subfigure{\includegraphics[width=0.42\columnwidth]{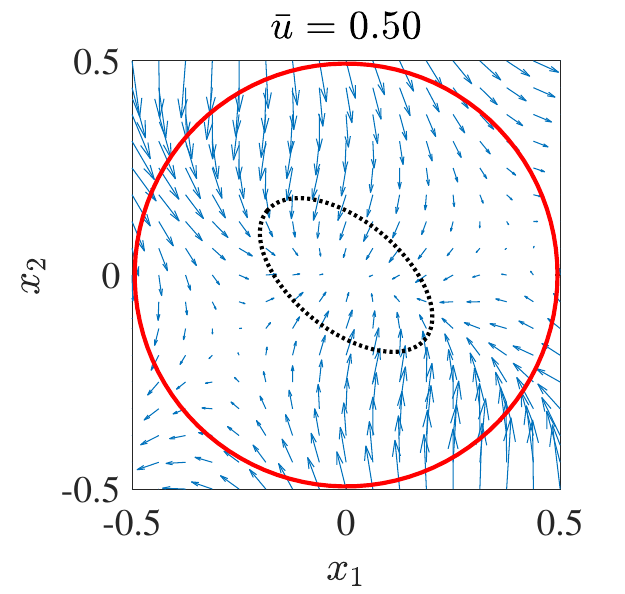}} 
    \subfigure{\includegraphics[width=0.42\columnwidth]{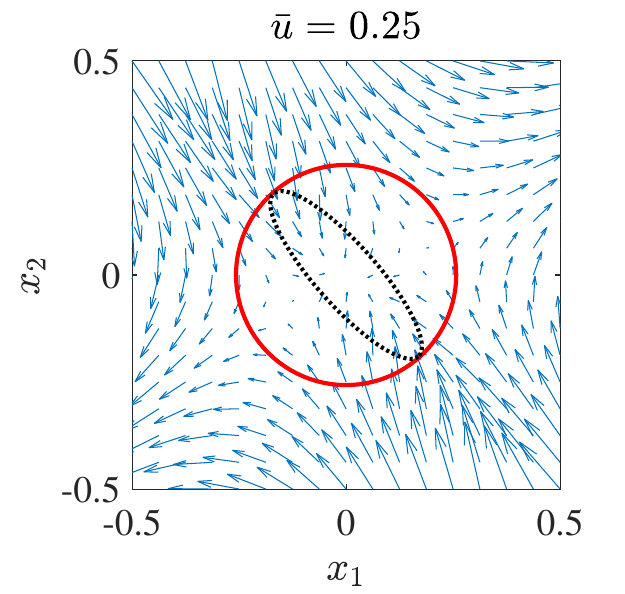}} 
    \vspace{-3mm}
    \caption{Estimate of ROA, $\mathcal{B}_{\bar{r}_0} \cap \mathcal{E}_{\bar{P},1}$, for different saturation levels.
    The red solid ball is $\mathcal{B}_{\bar{r}_0}$; 
    the black dashed ellipsoid is $\mathcal{E}_{\bar{P},1}$;
    the background is the phase portrait regarding the closed-loop system $x^{+} = A(x)x + B(x) \sat(Kx)$.}
    \label{fig:data_satur_u}
    \vspace{-3mm}
\end{figure}

\section{Physical Experiment Details}\label{apx:ExpDetails}

For the physical experiment in Section~\ref{sec:physical}, the excitation input $u_d$ and the corresponding state responses, namely the Euler angles and angular velocities, are depicted in Figure~\ref{fig:exp_data}.
The input $u_d$ in \eqref{eq:nSYSrewrite_distur} is generated by a controller with degraded performance.

\begin{figure}[b]
    \centering
    \includegraphics[width=0.9\linewidth]{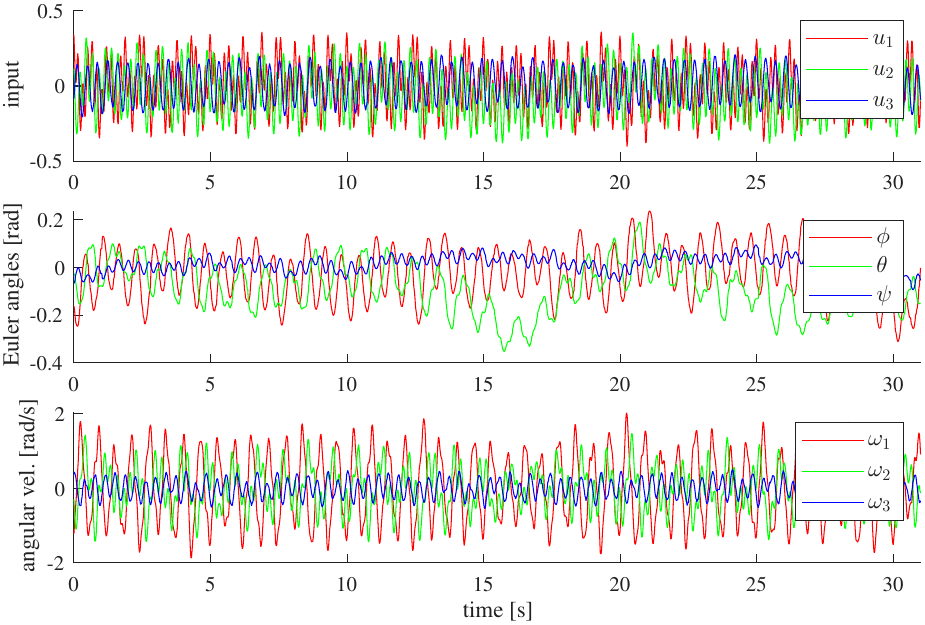}  
    \vspace{-2mm}
    \caption{Experimental data. 
    First subfigure is the experimental input $u_d$ as in \eqref{eq:nSYSrewrite_distur}, which is generated by a controller with degraded performance.
    The second and third subfigures show the corresponding state responses, namely the Euler angles and angular velocities, respectively.  
    }
    \label{fig:exp_data}
    \vspace{-4mm}
\end{figure}

Next, we derive the noise bound $\Theta$ used for physical experiment in Section~\ref{sec:physical}.
Let
$\mathrm{x}=(\eta,\omega)
=[\phi,\ \theta,\ \psi,\ \omega_1,\ \omega_2,\ \omega_3]^{\top}$
denote the true state,
and let
$x=\mathrm{x}+e$
denote its measurement, where $e$ is the measurement noise.
The continuous-time dynamics are given by
$\dot{\mathrm{x}}=f_{\mathrm{c}}(\mathrm{x},u)$.
Using forward Euler discretization with sampling time $h=0.002$ gives
$\mathrm{x}^{+}=f_{\mathrm{d}}(\mathrm{x},u)$, and 
$f_{\mathrm{d}}(\mathrm{x},u) := \mathrm{x}+h f_{\mathrm{c}}(\mathrm{x},u)$.
Because of the measurement noise $e$, the measured state satisfies
$x^{+}=f_{\mathrm{d}}(x,u)+d$, where
\begin{align}
    d(k)
    &=
    x(k+1)-f_{\mathrm{d}}\bigl(x(k),u(k)\bigr) \label{eq:apx_noise_exact}
    \\
    &=
    e(k+1)
    -
    \Bigl[
    f_{\mathrm{d}}\bigl(\mathrm{x}(k)+e(k),u(k)\bigr)
    -
    f_{\mathrm{d}}\bigl(\mathrm{x}(k),u(k)\bigr)
    \Bigr]. \notag
\end{align}
A first-order Taylor expansion of $f_{\mathrm{d}}$ about $\mathrm{x}$ yields
\begin{equation*}
    f_{\mathrm{d}}(\mathrm{x}+e,u)
    -
    f_{\mathrm{d}}(\mathrm{x},u)
    =
    \frac{\partial f_{\mathrm{d}}}{\partial \mathrm{x}}
    (\mathrm{x},u)e
    +
    O\bigl(|e|^{2}\bigr).
\end{equation*}
Hence, \eqref{eq:apx_noise_exact} can be approximated as
\begin{align}
    d(k)
    &\simeq
    e(k+1)
    -
    \frac{\partial f_{\mathrm{d}}}{\partial \mathrm{x}}
    \bigl(\mathrm{x}(k),u(k)\bigr)e(k)
    \notag\\
    &=
    \underbrace{e(k+1)-e(k)}_{=:d_{\mathrm{diff}}(k)}
    -
    \underbrace{
    h
    \frac{\partial f_{\mathrm{c}}}{\partial \mathrm{x}}
    \bigl(\mathrm{x}(k),u(k)\bigr)e(k)
    }_{=:d_{\mathrm{cpl}}(k)} .
    \label{eq:apx_noise_decomposition}
\end{align}
Here, $d_{\mathrm{diff}}$ denotes the noise-difference term;
$d_{\mathrm{cpl}}$ captures the coupling between the
measurement noise and the system dynamics.

\begin{table}[t]
    \centering
    \caption{Measurement-noise parameters.}
    \label{tab:noise_parameters}
    \scriptsize
    \setlength{\tabcolsep}{2.5pt}
    \renewcommand{\arraystretch}{1.05}
    \begin{tabular}{ccccc}
        \toprule
        $i$ & $b_i$ & $\sigma_i$ & $\rho_i$ & Source \\
        \midrule
        $\phi$
        & $0$
        & $3.53553{\times}10^{-3}$ rad
        & $0.977529$
        & EKF2 accuracy
        \\
        $\theta$
        & $0$
        & $3.53553{\times}10^{-3}$ rad
        & $0.977529$
        & EKF2 accuracy
        \\
        $\psi$
        & $0$
        & $1.20208{\times}10^{-2}$ rad
        & $0.977529$
        & EKF2 accuracy
        \\
        $\omega_1$
        & $4.35000{\times}10^{-3}$ rad/s
        & $1.24412{\times}10^{-3}$ rad/s
        & $0.081003$
        & IMU specification
        \\
        $\omega_2$
        & $4.35000{\times}10^{-3}$ rad/s
        & $1.24412{\times}10^{-3}$ rad/s
        & $0.081003$
        & IMU specification
        \\
        $\omega_3$
        & $4.35000{\times}10^{-3}$ rad/s
        & $1.24412{\times}10^{-3}$ rad/s
        & $0.081003$
        & IMU specification
        \\
        \bottomrule
    \end{tabular}
\end{table}

Denote the channel index by
$i\in\{\phi,\theta,\psi,\omega_1,\omega_2,\omega_3\}$.
Each $e_i$ consists of a constant turn-on bias with standard deviation
$b_i$ and a zero-mean fluctuation with standard deviation $\sigma_i$
and lag-one autocorrelation coefficient $\rho_i$,
evaluated at the corresponding source sampling rate.
The parameters used in the experiment are listed in
Table~\ref{tab:noise_parameters}.

\begin{table}[t]
    \centering
    \caption{Statistics of the noise-difference term.}
    \label{tab:noise_diff}
    \scriptsize
    \setlength{\tabcolsep}{2.5pt}
    \renewcommand{\arraystretch}{1.05}
    \begin{tabular}{ccccc}
        \toprule
        $i$
        & $\Delta t_{\mathrm{src},i}$
        & $s_i$
        & $\mathbb{E}[d_{\mathrm{diff},i}]$
        & $\operatorname{Var}[d_{\mathrm{diff},i}]$
        \\
        \midrule
        $\phi$
        & $11.3636$ ms
        & $5.681818$
        & $0$
        & $1.74015{\times}10^{-8}$
        \\
        $\theta$
        & $11.3636$ ms
        & $5.681818$
        & $0$
        & $1.74015{\times}10^{-8}$
        \\
        $\psi$
        & $11.3636$ ms
        & $5.681818$
        & $0$
        & $2.01161{\times}10^{-7}$
        \\
        $\omega_1$
        & $10.0000$ ms
        & $5.000000$
        & $0$
        & $1.13796{\times}10^{-7}$
        \\
        $\omega_2$
        & $10.0000$ ms
        & $5.000000$
        & $0$
        & $1.13796{\times}10^{-7}$
        \\
        $\omega_3$
        & $10.0000$ ms
        & $5.000000$
        & $0$
        & $1.13796{\times}10^{-7}$
        \\
        \bottomrule
    \end{tabular}
\end{table}

\begin{table}[b]
    \centering
    \caption{Statistics of the coupling term.}
    \label{tab:noise_coupling}
    \scriptsize
    \setlength{\tabcolsep}{4pt}
    \renewcommand{\arraystretch}{1.05}
    \begin{tabular}{ccc}
        \toprule
        $i$
        & $\sigma_{b,i}$
        & $\operatorname{Var}[d_{\mathrm{cpl},i}]$
        \\
        \midrule
        $\phi$
        & $1.017{\times}10^{-5}$
        & $8.142{\times}10^{-10}$
        \\
        $\theta$
        & $9.648{\times}10^{-6}$
        & $4.137{\times}10^{-10}$
        \\
        $\psi$
        & $1.099{\times}10^{-5}$
        & $6.945{\times}10^{-10}$
        \\
        $\omega_1$
        & $2.315{\times}10^{-5}$
        & $4.385{\times}10^{-11}$
        \\
        $\omega_2$
        & $2.315{\times}10^{-5}$
        & $4.385{\times}10^{-11}$
        \\
        $\omega_3$
        & $0$
        & $0$
        \\
        \bottomrule
    \end{tabular}
\end{table}

The attitude channels are obtained from the onboard estimator rather
than measured directly.
Their standard deviations are therefore taken from the converged EKF2
accuracy, and no constant bias is included.
The yaw angle $\psi$ is less accurate than $\phi$ and $\theta$
because its heading estimate relies on the magnetometer.
For the body-rate channels, $b_i$ corresponds to the turn-on bias of
the gyroscope, which is assumed to remain constant during each
data-collection experiment.

Each channel $i$ is recorded with its own sampling period
$\Delta t_{\mathrm{src},i}$ and resampled onto the common grid with
sampling time $h$ using linear interpolation.
Thus, the change in $e_i$ between two consecutive source samples is
distributed over $s_i:={\Delta t_{\mathrm{src},i}} / {h}$ intervals on the common grid.
Since two consecutive source samples have the same mean,
$d_{\mathrm{diff},i}=e_i(k+1)-e_i(k)$ has zero mean; its variance is
\begin{equation}
    \operatorname{Var}\bigl[d_{\mathrm{diff},i}\bigr]
    =
    \frac{2\sigma_i^2(1-\rho_i)}{s_i^2}.
    \label{eq:apx_noise_diff_var}
\end{equation}
Evaluating \eqref{eq:apx_noise_diff_var} for each channel gives
Table~\ref{tab:noise_diff}.

The coupling term in \eqref{eq:apx_noise_decomposition} can be written as
\begin{equation*}
    d_{\mathrm{cpl}}
    =
    h
    \frac{\partial f_{\mathrm{c}}}{\partial \mathrm{x}}e
    =:
    G_{\mathrm{cpl}}e,
\end{equation*}
whose $i$-th component is $d_{\mathrm{cpl},i} = \sum_j \bigl(G_{\mathrm{cpl}}\bigr)_{ij}e_j$.
Assuming that the six components of $e$ are mutually uncorrelated,
the variance of the constant offset and that of the zero-mean
fluctuation are, respectively,
\begin{equation}\label{eq:apx_bias_cpl_var}
    \sigma_{b,i}^{2} 
    = 
    \sum_j \bigl(G_{\mathrm{cpl}}\bigr)_{ij}^{2} b_j^{2}, 
    \quad
    \operatorname{Var}\bigl[d_{\mathrm{cpl},i}\bigr] 
    = 
    \sum_j \bigl(G_{\mathrm{cpl}}\bigr)_{ij}^{2} \sigma_j^{2}.
\end{equation}

Over the recorded data,
$|\eta_i|\leq0.5~\mathrm{rad}$ and
$|\omega_i|\leq2.1~\mathrm{rad/s}$.
Let $\bar{G}_{ij}$ denote the largest value of
$|(G_{\mathrm{cpl}})_{ij}|$ over this range.
We obtain
\begin{equation*}
\begin{aligned}
    \bar{G}
    =
    10^{-3}
    \smat{
        3.114 & 7.400 & 0 & 2.000 & 0.524 & 1.093 \\
        5.699 & 0     & 0 & 0     & 2.000 & 0.959 \\
        6.494 & 3.548 & 0 & 0     & 1.093 & 2.279 \\
        0     & 0     & 0 & 0     & 3.764 & 3.764 \\
        0     & 0     & 0 & 3.764 & 0     & 3.764 \\
        0     & 0     & 0 & 0     & 0     & 0
    },
\end{aligned}
\end{equation*}
where the rows and columns are ordered as
$\phi,\theta,\psi,\omega_1,\omega_2,\omega_3$.
The evaluation of $\bar{G}$ requires the inertia matrix,
for which we use the estimate
$J = \blkdiag(0.029125,\ 0.029125,\ 0.055225) \ \mathrm{kg\,m^2}$.
As shown below, $d_{\mathrm{cpl}}$ contributes only a small fraction
of the total noise energy.
Therefore, moderate errors in the estimated inertia have little effect
on the resulting noise bound.
Substituting $\bar{G}$ for $G_{\mathrm{cpl}}$ in \eqref{eq:apx_bias_cpl_var} gives Table~\ref{tab:noise_coupling}.

From \eqref{eq:apx_noise_decomposition},
$d_i=d_{\mathrm{diff},i}-d_{\mathrm{cpl},i}$.
Assuming that the two terms are uncorrelated, and noting that only
$d_{\mathrm{cpl},i}$ contains a constant offset, we obtain
\begin{equation}
    \operatorname{Var}[d_i]
    =
    \operatorname{Var}[d_{\mathrm{diff},i}]
    +
    \operatorname{Var}[d_{\mathrm{cpl},i}],
    \label{eq:apx_total_noise_var}
\end{equation}
which gives Table~\ref{tab:noise_total}.

The energy of channel $i$ is defined as 
$E_i:=\sum_{k=0}^{T-1}d_i(k)^2$.
Let
$d_i(k)=c_i+\delta_i(k)$, 
where $c_i$ is the constant offset satisfying
$\mathbb{E}[c_i^2]=\sigma_{b,i}^2$,
and $\delta_i(k)$ is a zero-mean fluctuation satisfying
$\mathbb{E}[\delta_i^2]=\operatorname{Var}[d_i]$.
Then
\begin{align}
    E_i
    &=
    T c_i^2
    +
    2c_i\sum_{k=0}^{T-1}\delta_i(k)
    +
    \sum_{k=0}^{T-1}\delta_i(k)^2 .
    \label{eq:apx_channel_energy}
\end{align}
The middle term has zero mean.
If $\delta_i(k)$ were independent Gaussian samples,
the last term in \eqref{eq:apx_channel_energy}, normalized by
$\operatorname{Var}[d_i]$, would follow a $\chi_T^2$ distribution,
with mean $T$ and variance $2T$.
Because linear interpolation introduces correlation between consecutive
samples, we approximate the number of independent samples by
$T_{\mathrm{eff},i}:={T} / {s_i}$.
The mean energy remains $T\operatorname{Var}[d_i]$, whereas its
relative standard deviation becomes
$\sqrt{2/T_{\mathrm{eff},i}}$.
Using a $\kappa$-standard-deviation upper bound for the fluctuation
energy and replacing $c_i^2$ by $\sigma_{b,i}^2$ gives
\begin{equation}
    \bar{E}_i
    =
    T\sigma_{b,i}^2
    +
    T\operatorname{Var}[d_i]
    \left(
        1+\kappa\sqrt{\tfrac{2}{T_{\mathrm{eff},i}}}
    \right),
    \quad
    T_{\mathrm{eff},i}
    =
    \tfrac{T}{s_i}.
    \label{eq:apx_energy_bound}
\end{equation}
With $T=15500$ and $\kappa=3$,
\eqref{eq:apx_energy_bound} gives the values in
Table~\ref{tab:noise_energy}.

Using the $3\sigma$ approximation for the chi-square energy term
gives an approximate coverage of $99.83\%$ for each channel.
A union bound over the six channels then gives a joint confidence of
approximately $99.0\%$.
Accordingly,
\begin{equation*}
    D_0D_0^{\top}
    \preceq
    \operatorname{tr}\bigl(D_0D_0^{\top}\bigr)I_6
    \preceq
    9.705\times10^{-3}I_6 .
\end{equation*}
We therefore take $\Theta=0.01I_6$ for Assumption~\ref{ass:DisturbanceEnergy} in Section~\ref{sec:physical}.

\begin{table}[t]
    \centering
    \caption{Statistics of the total discrete-time noise.}
    \label{tab:noise_total}
    \scriptsize
    \setlength{\tabcolsep}{5pt}
    \renewcommand{\arraystretch}{1.05}
    \begin{tabular}{ccc}
        \toprule
        $i$
        & $\sigma_{b,i}$
        & $\operatorname{Var}[d_i]$
        \\
        \midrule
        $\phi$
        & $1.017{\times}10^{-5}$
        & $1.821{\times}10^{-8}$
        \\
        $\theta$
        & $9.648{\times}10^{-6}$
        & $1.781{\times}10^{-8}$
        \\
        $\psi$
        & $1.099{\times}10^{-5}$
        & $2.019{\times}10^{-7}$
        \\
        $\omega_1$
        & $2.315{\times}10^{-5}$
        & $1.138{\times}10^{-7}$
        \\
        $\omega_2$
        & $2.315{\times}10^{-5}$
        & $1.138{\times}10^{-7}$
        \\
        $\omega_3$
        & $0$
        & $1.138{\times}10^{-7}$
        \\
        \bottomrule
    \end{tabular}
\end{table}

\begin{table}[t]
    \centering
    \caption{Channel-wise noise-energy bounds.}
    \label{tab:noise_energy}
    \scriptsize
    \setlength{\tabcolsep}{2.5pt}
    \renewcommand{\arraystretch}{1.05}
    \begin{tabular}{ccccc}
        \toprule
        $i$
        & $T_{\mathrm{eff},i}$
        & $T\sigma_{b,i}^2$
        & \shortstack{$T\operatorname{Var}[d_i]$\\
        $\times(1+\kappa\sqrt{2/T_{\mathrm{eff},i}})$}
        & $\bar{E}_i$
        \\
        \midrule
        $\phi$
        & $2728$
        & $1.603{\times}10^{-6}$
        & $3.053{\times}10^{-4}$
        & $3.069{\times}10^{-4}$
        \\
        $\theta$
        & $2728$
        & $1.443{\times}10^{-6}$
        & $2.985{\times}10^{-4}$
        & $3.000{\times}10^{-4}$
        \\
        $\psi$
        & $2728$
        & $1.872{\times}10^{-6}$
        & $3.384{\times}10^{-3}$
        & $3.385{\times}10^{-3}$
        \\
        $\omega_1$
        & $3100$
        & $8.307{\times}10^{-6}$
        & $1.899{\times}10^{-3}$
        & $1.907{\times}10^{-3}$
        \\
        $\omega_2$
        & $3100$
        & $8.307{\times}10^{-6}$
        & $1.899{\times}10^{-3}$
        & $1.907{\times}10^{-3}$
        \\
        $\omega_3$
        & $3100$
        & $0$
        & $1.898{\times}10^{-3}$
        & $1.898{\times}10^{-3}$
        \\
        \bottomrule
    \end{tabular}
\end{table}

\end{document}